\documentclass[prl,nobalancelastpage,twocolumn,nolongbibliography,preprintnumbers]{revtex4-2}

\usepackage{color,amsthm,amsmath,amsxtra,amsfonts,dsfont,graphicx,bm,amssymb}
\usepackage[colorlinks=true,linkcolor=blue, citecolor=blue, urlcolor=blue, bookmarks]{hyperref}
\usepackage{centernot}
\usepackage[dvipsnames]{xcolor}
\usepackage{graphicx}
\usepackage{mathtools}
\usepackage{tikz}
\usepackage{braket}
\usepackage{multirow}
\usepackage{makecell}
\usepackage{dsfont}
\usepackage{outlines}
\usepackage{verbatim}

\newcommand{\Tr}{\textrm{Tr}}

\newtheorem{theorem}{Theorem}
\newtheorem{lemma}{Lemma}

\newcommand{\prlsection}[1]{{\em {#1}---~}}

\begin{document}
\preprint{MIT-CTP/6103}

\title{Robust Hamiltonian engineering with subensemble control}

\author{Wenjie Gong}
\email{wgong@mit.edu}
\author{Matteo Votto}
\email{votto@mit.edu}
\author{Soonwon Choi}
\email{soonwon@mit.edu}
\affiliation{Center for Theoretical Physics---a Leinweber Institute, Massachusetts Institute of Technology, Cambridge, MA 02139, USA}

\begin{abstract}
    We present a robust protocol to reshape interactions in a spin ensemble based on global control pulse sequences applied to multiple subensembles in parallel.  
    This setting arises naturally from ensembles of solid-state defects or multi-species atomic arrays.
    We show that it is provably computationally hard to find pulse sequences that simultaneously engineer interactions both within and between subensembles.
    Despite its formal hardness, we identify a set of necessary or sufficient conditions under which one can synthesize a target Hamiltonian from the native one.
    Moreover, we introduce efficient numerical strategies for designing pulse sequences that engineer target Hamiltonians robust against common control imperfections.
    As a specific application, we discuss the robust generation of two-mode spin squeezing in dual-species atomic ensembles, which is a challenging task without subensemble control.
    Our results provide a practical toolbox to design novel quantum simulation and sensing experiments with minimal control overhead.
\end{abstract}

\maketitle


\begin{figure}[t]
    \centering
    \includegraphics[width=\linewidth]{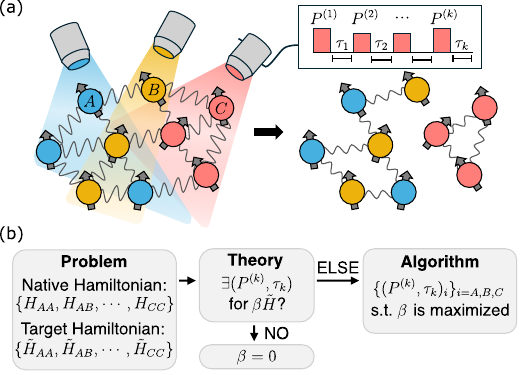}
    \caption{Hamiltonian engineering with subensemble control.
    (a) We consider interacting spin systems, where one can apply global control pulses to each of $n$ subensembles independent. As example, $n=3$ shown : $A$ (blue), $B$ (yellow) and $C$ (red). 
    We study how to engineer interactions within and between them in this setup.
    (b) Given a set of native interactions, we can determine the feasibility of engineering a target Hamiltonian.
    Moreover, we provide algorithms to design pulse sequences that optimize the strength of the resultant Hamiltonian.}
    \label{fig:protocol}
\end{figure}
The experimental realization of many-body Hamiltonians enables several applications, from the simulation of strongly correlated quantum matter~\cite{Blatt2012, Gross2017, Browaeys2020rydberg,Kjaergaard2020, Schafer2020, Altman2021, Morgado2021, Monroe2021, Burkard2023} to advanced sensing schemes~\cite{Degen2017, PezzeSmerzi2018,YeZoller2024,Montenegro_2025}.
While one strategy is to directly assemble quantum degrees of freedom tailored for specific target Hamiltonian, it is more desirable to engineer a wide range of different interactions within a single experimental setup in a programmable manner. 
A successful solution is \textit{pulsed Hamiltonian engineering}, where native interactions are reshaped by driving quantum spin degrees of freedom periodically with external field `pulses'~\cite{Leung01HamEng, Ajoy13HamEng, Wei18HamEng, Choi17qudit, Choi2020robust, Geier2021HamEng, Scholl22Rydberg, Miller2024Hameng_molecules, Zhou2024robustqudits, chen2026robust, Dodd02HamEng, Votto2024walsh, bassler2025robust}.
In principle, site-resolved pulse sequences enable the engineering of arbitrary Hamiltonians~\cite{Wocjan02, Dodd02HamEng, Bennett02HamEng, nielsen2002qudits, Bremner2005qudits, Votto2024walsh, bassler2025robust}.
However, to achieve such control is hard in many different physical platforms, such as solid-state spin systems~\cite{Doherty2013nv, Kucsko2018nv, Peng2021nv,Cheng2025nv, Cambria2025nv, Cakan2025hbn}, or atomic ensembles due to limited spatial resolution of controls in existing experiments~\cite{Raimond2001cavity,Braverman19cavity_squeezing, Muniz2020cavity,Li2022cavity, Li2023cavity, luo2025cavity}.
For this reason, most experimental implementations of Hamiltonian engineering have considered global control, where a system is driven by homogeneous fields~\cite{Choi2020robust,Geier2021HamEng,Scholl22Rydberg,kranzl2023ionsQ,kranzl2023ions,Miller2024Hameng_molecules,gao2025nv}, which severely constrains the realizable interactions~\cite{masanes2002global, Choi17qudit}.
A largely unexplored but experimentally demonstrated regime is \textit{subensemble control}, 
where a small number of subgroups of degrees of freedom are driven independently, see Fig.~\ref{fig:protocol}(a).
This has been demonstrated by parallel addressing of different crystalline orientations in NV centers~\cite{Kucsko2018nv,Cheng2025nv,Cambria2025nv}, using detuning masks in optical tweezers~\cite{gong2023,Bornet2024detuning} or multiple atomic species in trapped ion chains or atomic ensembles~\cite{anand2024dual,white2026dual,Bruzewicz2019ions,Moses2023ions}.
While subensemble control seems more powerful than global control, the extent of its advantage in pulsed Hamiltonian engineering remains unclear.
In this Letter, we show that the independent control of $n > 1$ subensembles of qudits allows one to engineer a much larger class of Hamiltonians than with global pulse sequences, and provide a protocol to design robust sequences suitable for near-term experiments.
First, we prove that determining whether or not a target Hamiltonian can be engineered from a native one with subensemble control is computationally hard in general.
Nevertheless, we provide a computable necessary condition for engineerability which is tight in some practical cases, as well as a sufficient one, uncovering a large class of realizable dynamics.
To improve on these results in practical settings, we provide controlled numerical methods to design pulse sequences that optimize the strength of target interactions.
Finally, we explicitly analyze the effects of common pulse imperfections, and develop an analytical framework to obtain robust sequences.
By doing so, we generalize previous results in related settings that lie at the core of successful experimental implementations of pulsed Hamiltonian engineering~\cite{Choi2020robust,kranzl2023ionsQ,kranzl2023ions,gao2025nv}.
We illustrate the utility of our approach by showing how two-mode spin squeezing~\cite{Gessner2020squeezing2,Kitzinger2020squeezing2,Mamaev2025squeezing2,kaubruegger2026squeezing2,Chu2026subensemble} can be robustly generated with atomic systems in either optical cavities or tweezers with control on two subensembles.


\prlsection{Engineerable Hamiltonians} 
For simplicity, let us first consider a system of $N$ qudits (\textit{i.e.} $d$-level systems) grouped in $n=2$ subensembles, $A$ and $B$, interacting under a native $2$-local Hamiltonian $H = H_{AA} + H_{BB} + H_{AB}$, with
\begin{equation}\label{eq:hamnative}
    H_{AA} = \sum_{(i < j)\in A} J_{ij} h_{ij}^{AA} ~,~ H_{AB} = \sum_{i\in A, j\in B} J_{ij} h_{ij}^{AB} \ ,
\end{equation}
and similarly for $H_{BB}$, where $J_{ij}$ are interaction strengths. Here,
\begin{equation}
\label{eq:int_mat}
     h_{ij}^{AA} = \sum_{\mu, \nu} g_{\mu\nu}^{AA} \lambda_i^{\mu}\otimes\lambda_j^{\nu} \ ,
\end{equation}
and equivalently for $ h_{ij}^{AB}$ and $ h_{ij}^{BB}$, where $\{\lambda^{\mu}\}_{\mu = 1}^{d^2-1}$ are generalized Gell-Mann operators, \textit{i.e.} qudit generalization of Pauli matrices~\cite{Choi17qudit}, and $g_{\mu\nu}$ are interaction matrices, uniquely specifying the type of interactions among particles. 
For example, for $d=2$, $g \propto \text{diag}(1, 1, -2)$ for dipolar interactions and $g \propto \text{diag}(1, 1, 1)$ for Heisenberg interactions.
We assume $g_{\mu \nu}$ is position-independent, and all distance-dependent interaction strengths are captured by $J_{ij}$. 
We will comment later on how to generalize our results to arbitrary $n$.
In pulsed Hamiltonian engineering, one engineers the dynamics of an effective Hamiltonian $\tilde{H}$ by interspersing the native dynamics under $H$ with a sequence of external field pulses $P^{(k)} = (\bigotimes_{i \in A} p_{A,i}^{(k)})\otimes( \bigotimes_{j \in B} p_{B,j}^{(k)})$ with $k\in \{1, \dots, K\}$.
Applying a pulse sequence $\{P^{(k)}\}$ with period $\tau = \sum_k \tau_k$, with $\prod_k P^{(k)} = I$ being the identity, and where $\{\tau_k\}$ are the time-spacings between the pulses, we obtain the dynamics
\begin{equation}\label{eq:floqu}
    \tilde{U}(\tau) = e^{-i H \tau_K} P^{(K)} e^{-i H \tau_{K-1}}P^{(K-1)}\cdots e^{-iH\tau_1}P^{(1)} \ .
\end{equation}
We note that this expression assumes an idealization: in practice, each pulse takes a finite time duration during which the qudits still interact under $H$.
We will relax this approximation later when considering pulse imperfections.
Using the Floquet-Magnus expansion~\cite{magnus1954exponential}, the unitary dynamics at integer multiples of $\tau$ is described by $(\tilde{U}(\tau))^{T/\tau} = e^{-i\tilde{H}T} + O(\tau T)$, where the effective Hamiltonian is $\tilde{H} = \tilde{H}_{AA} + \tilde{H}_{BB} + \tilde{H}_{AB}$, with $\tilde{H}_{AA} = \sum_{i<j}J_{ij}\tilde{h}_{ij}^{AA}$ and 
\begin{equation}
\label{eq:engH}
\begin{split}
    &\tilde{h}^{AA}_{ij} = \sum_k\frac{\tau_k}{\tau} (u^{(k)}_{A,i}\otimes u^{(k)}_{j,A})^{\dagger} h^{AA}_{ij} (u^{(k)}_{A,i}\otimes u^{(k)}_{j,A}) \ , \\
     &\tilde{h}^{AB}_{ij} = \sum_k\frac{\tau_k}{\tau} (u^{(k)}_{A,i}\otimes u^{(k)}_{j,B})^{\dagger} h^{AB}_{ij} (u^{(k)}_{A,i}\otimes u^{(k)}_{j,B}) \ ,
\end{split}
\end{equation}
where $u_{A,i}^{(k)} = \prod_{k'=1}^k p^{(k)}_{A,i}$ are toggling-frame unitaries, and equivalently for the remainder.
Noting that pulsed Hamiltonian engineering with local control is capable of generating arbitrary interactions, it may appear as if two-subensemble control could enable arbitrary inter-subensemble interactions $\tilde{h}_{ij}^{AB}$, up to a constant factor~\cite{Bennett02HamEng}.
This is not the case; the pulse sequence to engineer inter-subensemble interactions may affect interactions within each subensemble, making engineering Hamiltonians a nontrivial constrained problem.
For instance, let us consider $h_{ij} = X_i X_j + Y_i Y_j$ for every pair of qubits both within and between subensembles, where $X,Y,Z$ are Pauli matrices.
As we will prove later, engineering $\tilde{h}_{ij}^{AB} = \gamma Z_i Z_j$ while keeping $h^{AA}_{ij}$ and $h^{BB}_{ij}$ unchanged is impossible for $\gamma \neq 0$.
Such constraints make determining the engineerability of $\tilde{H}$ from a given $H$ computationally hard.
In particular, we prove that this is harder than deciding the separability of a two-qudit mixed state, known to be NP-hard in the dimension of the qudit $d$~\cite{SM, Horodecki1996entanglement,Gurvits2003locc,Fano1957, Fano1983, SP_2017}, as well as deciding membership in the polytope class $CUT^{\pm}_{n}$, which is NP-hard in the number of subensembles $n$~\cite{SM, Pitowsky1991}.
For $d=n=2$, the Hamiltonian engineering problem is closely related to determining whether two two-qubit mixed states are related by mixture of local unitaries, a currently unsolved problem~\cite{SM,Nielsen1999entanglement,Makhlin2002LU,li2011entanglement,lin2024entanglement,zhou2024LU}.
Despite this formal hardness, it is possible to rule out the engineerability of several Hamiltonians under the simultaneous control of subensembles.
\begin{theorem}
\label{th:necess}
    Given native interaction matrices $\{g^{AA}, g^{AB}, g^{BB}\}$, as defined in Eq.~\eqref{eq:int_mat}, and where $g^{AA}$ and $g^{BB}$ are symmetric, one can engineer target interaction matrices $\{\tilde{g}^{AA}, \tilde{g}^{AB}, \tilde{g}^{BB}\}$ where $\tilde{g}^{AA}$ and $\tilde{g}^{BB}$ are symmetric using subensemble control only if 
\begin{equation}
\label{eq:maj}
    \sum_{i=1}^{\ell}\lambda_i^{\downarrow}(\tilde{G})\leq \sum_{i=1}^{\ell}\lambda_i^{\downarrow}(G) \ ,
\end{equation}
for every $\ell\leq d^2-1$ and where 
\begin{equation}
\label{eq:coeffm}
    G = \begin{pmatrix}
        g^{AA} & g^{AB} \\ 
        (g^{AB})^T & g^{BB}
    \end{pmatrix} \ , 
\end{equation}
and similarly for $\tilde{G}$, and $\lambda_i^{\downarrow}(\cdot)$ is the $i$-th eigenvalue of a matrix in decreasing order.
\end{theorem}
This follows from the fact that individual pulses act on $G$ as orthogonal rotations, and $\tilde{G}$ is a convex combination of them, as we discuss in the End Matter~\cite{masanes2002global,Choi2003convex,bhatia2013matrix}.
Thm.~\ref{th:necess} arises solely from the constrained structure of the problem, and needs to be satisfied together with conditions for global and local control on intra- and inter-subensemble interactions.
Applying this result to the example discussed above, we obtain $\sum_{i=1}^{5}\lambda_i^{\downarrow}(\tilde{G})=4+|\gamma|$ and $\sum_{i=1}^{5}\lambda_i^{\downarrow}(G)=4$, hence $\gamma = 0$.
Additionally, we prove the engineerability of a large class of Hamiltonians.
\begin{theorem}
\label{th:suff}
    Given native interaction matrices $\{g^{AA}, g^{AB}, g^{BB}\}$ such that $\text{tr}[g^{AA}] = \text{tr}[g^{BB}] = 0$, one can engineer using subensemble control any triple $\{\beta \tilde{g}^{AA}, \beta \tilde{g}^{AB}, \beta \tilde{g}^{BB}\}$ up to a rescaling factor $\beta>0$, as long as $\text{tr}[\tilde{g}^{AA}] = \text{tr}[\tilde{g}^{BB}] = 0$.
\end{theorem}
We prove this in the End Matter by constructing pulse sequences that selectively decouple interactions either between or within subensembles while not affecting the others. 
This is possible only when the interaction matrices are traceless~\cite{dur2000design, Choi17qudit, SM, Mele2024haar}.
Then, these pulse sequences can be composed to realize any target interaction. 
This implies that the constraints of the Hamiltonian engineering problem solely arise from the trace invariance of $g^{AA,BB}$, \textit{i.e.} that the strength of the component proportional to the generalized Heisenberg Hamiltonian $h_{ij}^0= \sum_{\mu}\lambda^{\mu}_i\lambda^{\mu}_j$ cannot be changed.
Therefore, Thm.~\ref{th:suff} gives not only a formal guarantee, but also suggests a reformulation of the Hamiltonian engineering problem:
given $\tilde{H}$, one can separate the intra-subensemble Heisenberg component, and design the pulse sequence maximizing the rescaling factor $\beta$ of the remainder.
Our results extend naturally to $n > 2$ subensembles.
The decomposition in Eqs.~\eqref{eq:hamnative}-\eqref{eq:engH} generalizes to $n$ intra-subensemble interaction terms and $\frac{n(n-1)}{2}$ inter-subensemble terms.
For $n<N$, intra-subensemble interactions constrain the problem, and the Hamiltonian engineering problem remains NP-hard in $d$.
The necessary condition of Thm.~\ref{th:necess} extends directly to any subset of subensembles, resulting in a total of $O(2^n)$ additional conditions.
Moreover, Thm.~\ref{th:suff} generalizes to arbitrary $n$: selective decoupling can be now used to independently control all intra- and inter-subensemble couplings, implying that any such target Hamiltonian can be engineered up to an overall rescaling factor.
%


\prlsection{Numerical pulse sequence design}
We now provide controlled numerical techniques to design pulse sequences that engineer target interactions given subensemble control.
Suppose we want to engineer a set of intra- and inter-subensemble interactions $\{\tilde{h}^{\text{intra}, a}, \tilde{h}^{\text{inter}, b} \}$, labeled by the indices $a$ and $b$, that we can express as 
\begin{equation}
\label{eq:opt_problem}
    \tilde{h}^{\text{intra}, a} = \alpha_a h^0 + \beta_a \bar{h}^{\text{intra}, a} ~ , ~ \tilde{h}^{\text{inter}, b} = \gamma_b\bar{h}^{\text{inter}, b} \ ,
\end{equation} 
where $h^0_{ij}= \sum_{\mu}\lambda^{\mu}_i\lambda^{\mu}_j$, and $\alpha_a$, $\beta_a$, $\gamma_b$ are their associated coupling factors such that the matrices $\bar{h}^{\text{intra}, a}$ are traceless.
Then, Thm.~\ref{th:suff} tells us that, if we can engineer the interactions in Eq.~\eqref{eq:opt_problem} for some $\beta_a'\geq\beta_a$ and $\gamma'_b\geq\gamma_b$, we can also engineer them for $\beta_a, \gamma_b$. 
Therefore, pulse sequence design can be recasted as an optimization problem for the parameters $\{\beta_a, \gamma_b\}$.
We note that such an optimization will always be constrained by the NP-hardness of the problem. 
Nevertheless, we can formulate the search of a pulse sequence achieving $\text{max}(\beta_a, \gamma_b)$ in closed form by directly optimizing over the space of fixed $SU(d)$ pulses.
As $\tilde{h}$ is a convex combination of toggling-frame Hamiltonians (Eq.~\eqref{eq:engH}), we prove that any engineerable $\tilde{h}$ can be achieved with a pulse sequence of length $O((nd)^2)$~\cite{SM, bhatia2013matrix}.
By considering a sequence of this length, it suffices to optimize over a finite set of pulse parameters and intervals $\tau_k$.
Such formulation enables the direct implementation of established approximate optimization algorithms, such as differential evolution~\cite{Storn1997de, Das2011de}, which can provide solutions for modest values of $n$ and $d$.
In many cases, it is instead sufficient to restrict ourselves to a finite set of experimentally convenient pulses, enabling the use of linear programming (LP) algorithms~\cite{Bertsimas1997LP, Choi17qudit}.
While provably efficient, this approach requires a suitably expressive set of pulses.
More formally, if a solution to the Hamiltonian engineering problem with $\text{max}(\beta_a, \gamma_b)>0$ exists, it is desirable that LP always returns a sequence realizing some $\beta_a, \gamma_b>0$ using a fixed, finite pulse set.
We prove that a set of pulses $\{p^{(k)}\}$ satisfies this condition if it can generate a set of toggling-frame unitaries $\{u^{(k)}\}$ that forms a $4$-design~\cite{SM, Mele2024haar}.
This elucidates previous heuristic results: for instance, the icosahedral pulse sequence considered in Ref.~\cite{Ben2020iso} satisfies this property.
%


\begin{figure}[t]
    \centering
    \includegraphics[width=\linewidth]{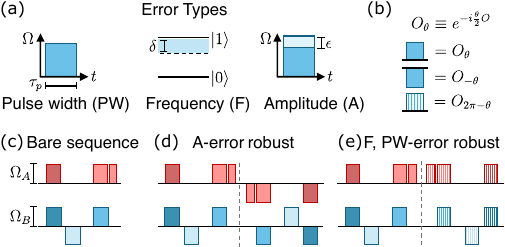}
    \caption{Robust pulse sequence design. 
    (a) We consider three common pulse-level coherent errors.
    (b) Our approach to robust pulse sequence design is based on introducing two kinds of complementary pulses for each pulse. Here, $O$ is a Pauli operator acting on the targeted two-level subspace. 
    (c-e) We depict the resulting robust pulse sequences using the notation introduced in (b).}
    \label{fig:robust}
\end{figure}
\prlsection{Robust pulse sequences}
So far, we have neglected experimental imperfections, which can drastically alter the resulting dynamics.
Indeed, the design of pulse sequences robust to relevant errors is crucial for their practical implementation~\cite{Choi2020robust,kranzl2023ionsQ,kranzl2023ions,Zhou2024robustqudits,Votto2024walsh,bassler2025robust}.
Here, we provide a universal method to make qudit pulse sequences robust to relevant errors, specifically amplitude (A), frequency (F), and pulse-width (PW) errors.
To achieve this, we will compile qudit pulses with a sequence of operations targeting two-level subspaces, as is necessary in most experiments~\cite{Zhou2024robustqudits}.
Hence, we express each elementary operation as a qubit pulse $O_{\theta} = e^{-i \frac{\theta}{2} O}$, where $O$ is a Pauli operator on the targeted two-level subspace, and $\theta > 0$.
A-errors act as multiplicative errors in the rotation angle $\theta \rightarrow \theta(1+\epsilon)$, with $\epsilon$ being a quasi-static error in the Rabi frequency of the driving field. 
Given a pulse sequence $\{O^{(k)}_{\theta^{(k)}}\}_{k=1}^{K}$, we prove that these errors can be mitigated by applying a reflected pulse sequence $\{O^{(K-k)}_{-\theta^{(K-k)}}\}_{k=1}^K$ every other period of the original one.
This results in a quadratic reduction of errors by applying an overall twice as long pulse sequence;
see Fig.~\ref{fig:robust}(d). 
F-errors manifest instead as a small detuning $\delta$ that adds a $Z$ component to pulses in the $x-y$ plane, \textit{i.e.} $\exp(-i\frac{\theta}{2}( X,Y+\delta  Z))$. 
Similarly, F-errors can be quadratically mitigated by alternating the original pulse sequence with another reflected as $\{O^{(K-k)}_{2\pi-\theta^{(K-k)}}\}_{k=1}^K$ (see Fig.~\ref{fig:robust}(e)). 
Composing these two strategies mitigates the two sources of errors simultaneously.
Finally, PW-errors are a result of the interaction between qudits under $H$ during the application of a pulse, that generally requires a finite time $\tau_p$.
When applying a pulse $O_{\theta}$, PW-errors contribute to the effective Hamiltonian with factors $\propto\frac{1}{\tau}\int_0^{\tau_p}(O_{\theta(t)}^{\dagger}\lambda^{\mu}_i O_{\theta(t)})(O_{\theta'(t)}'^{\dagger}\lambda^{\nu}_j O'_{\theta'(t)})dt$ with $\theta(0) = 0$, $\theta(\tau_p)=\theta$.
In general, this results in two qualitatively distinct error contributions.
The first contribution depends explicitly on $\theta$, and it can be corrected in first order in the same way as F-errors.
The second contribution is instead independent of the pulse shape $\theta(t)$; for $d=2$, where pulses act as $O_{\theta}^{\dagger}\lambda_i^{\mu} O_{\theta}=\cos(\theta)\lambda_i^{\mu} + \sin(\theta)\tilde{\lambda}_i^{\mu}$, it is proportional to $\tau_p(\lambda_i^{\mu}\lambda_j^{\nu}+\tilde{\lambda}_i^{\mu}\tilde{\lambda}_j^{\nu})$.
However, by leveraging the knowledge of its analytical form, it is possible to compensate for it by appropriately re-tuning the angles $\theta^{(k)}$ and the pulse distances $\tau_k$ numerically using LP algorithms~\cite{SM}.
These prescriptions are independent of the specific pulse sequence considered, and allow mitigation of several error sources in parallel.
In the next section, we provide practical examples of robust pulse sequences and study their performances.
%


\begin{figure}[t]
    \centering
    \includegraphics[width=\linewidth]{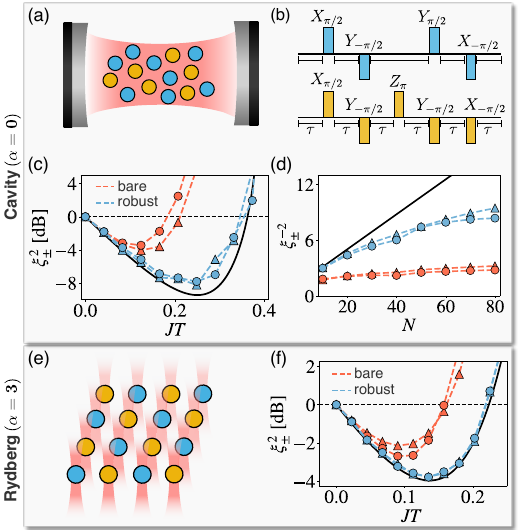}
    \caption{ 
    Two-mode spin squeezing via Hamiltonian engineering.
    We consider two possible experimental settings: (a-d) an atomic ensemble in an optical cavity and (e,f) dipolar Rydberg atoms in optical tweezers.
    In both cases, the desired target dynamics cannot be generated without subensemble control.
    (b) Alternating one-axis twisting interactions with the pulse sequence considered results in 2A2S twisting dynamics.
    (c) We study numerically the dynamics of the squeezing parameters $\xi_{\pm}(t)$ for an ensemble of $ N_{A,B} = 20$ atoms.
    We assume pulses with A- and F-error strengths $\epsilon=\delta=0.03$, comparing the bare pulse sequence (red lines, triangles for $\xi_+$ and circles for $\xi_-$) in (b) and its robust version (blue lines) with the exact dynamics of $\tilde{H}$ (black).
    (d) We also study the scaling of the maximal squeezing with $N$ in these cases.
    (e,f) We perform a similar analysis for the dynamics of the squeezing parameters for a $4\times 5$ tweezer array.
    In this case, we additionally include PW-errors with $\tau_p = \tau_k/4$.}
    \label{fig:squeezing}
\end{figure}

\prlsection{Application: multi-mode spin squeezing} 
We now apply our protocol to the preparation of multi-mode spin-squeezed states in atom arrays.
Let us consider a system of $N$ qubits in a collective spin state with $S=N/2$.
Spin-squeezed states are metrologically useful states characterized by a squeezing parameter $\xi<1$, defined as
\begin{equation}
    \xi^2 = \frac{N~\text{min}_{\perp}(\text{Var}(|\vec{S}_{\perp}|))}{|\langle \vec{S} \rangle |^2} \ , 
\end{equation}
where $\text{Var}(O) = \langle O^2\rangle-\langle O\rangle^2$, and the minimum is taken on the plane perpendicular to the collective spin $\langle \vec{S}\rangle = (\langle S^x \rangle, \langle S^y \rangle, \langle S^z \rangle)$~\cite{PezzeSmerzi2018}.
This results in a sensitivity $\delta\theta = \xi/\sqrt{N}$ to rotations $e^{-i\theta \vec{S}_{\perp}'}$, where $\vec{S}_{\perp}'$ is perpendicular to both $\langle\vec{S}\rangle$ and $\vec{S}_{\perp}$.

In the present work, we consider a multi-ensemble generalization of squeezed states, where it is possible to obtain a similar metrological enhancement for the simultaneous estimation of different parameters, such as the components of a vector field~\cite{Gessner2020squeezing2,Kitzinger2020squeezing2,Mamaev2025squeezing2,kaubruegger2026squeezing2,Chu2026subensemble}.
An example consists in evolving the state $ \ket{S^z=S}_A\otimes\ket{S^z=-S}_B$ of two collective spins $A$ and $B$ under the two-axis two-spin twisting (2A2S) Hamiltonian~\cite{Kitzinger2020squeezing2,Mamaev2025squeezing2,kaubruegger2026squeezing2}
\begin{equation}
\label{eq:2A2S}
    \tilde{H} = \tilde{J}(S^x_AS^x_B + S^y_AS^y_B) \ ,
\end{equation}
where $S^x$, $S^y$, $S^z$  are $x$, $y$ and $z$ spin operators.
This evolution is known to realize a class of two-mode spin-squeezed states with collective spin $S_-^z=  S^z_A-S^z_B$ and two squeezed components $S_{\pm}^{\perp} = (S_A^x\pm S_B^x)\pm(S_A^y-S_B^y)$, with associated parameters $\xi_{\pm}$.
In particular, both components can be squeezed to the Heisenberg limit, \textit{i.e.} $\xi_{\pm}^2\sim 1/N$, in time $O(\log(N)/N)$~\cite{Mamaev2025squeezing2, kaubruegger2026squeezing2}.
Crucially, engineering the Hamiltonian above requires subensemble control in existing experiments.
First, we consider the setting of an ensemble of $N_{A,B}$ atoms of two species $A,B$ in an optical cavity.
Optical cavities enable all-to-all interactions, resulting in one-axis twisting interactions $H = J (S_A^z +S^z_B)^2$ between two collective spins with $S = N_{A,B}/2$, treating each atom as a qubit~\cite{Li2022cavity,Li2023cavity}.
We can then engineer 2A2S interactions by addressing the species $A$ with the sequence $\{X_{\pi/2}, Y_{-\pi/2}, I, Y_{\pi/2}, X_{-\pi/2}, I\}$, and $B$ with $\{X_{\pi/2}, Y_{-\pi/2}, Z_{\pi}, Y_{-\pi/2}, X_{\pi/2}, I\}$, both with $\tau_k = \tau/6$.
We note that these sequences achieve the optimal $\tilde{J}=J/3$, saturating the necessary condition of Thm.~\ref{th:necess}.
We study spin squeezing generation by performing exact numerical simulations up to $N = 80$ atoms, $N_{A,B}=N/2$, and using $6$ cycles of the pulse sequence above to reach the optimal squeezing time~\footnote{We compute the optimal squeezing time by performing an exact numerical simulation of the 2A2S dynamics.}; see Fig.~\ref{fig:squeezing}(c,d).
We compare $\xi_{\pm}(t)$ for $N=40$ for the bare (red lines) and the robust pulse sequence (blue lines, designed using previous prescriptions), when explicitly including A- and F- errors~\footnote{PW-errors can be avoided by turning off $H$ during the application of the pulses, as routinely done in optical cavities.} with $\epsilon=\delta=0.03$, to the exact dynamics under $\tilde{H}$ (black line).
While we observe that pulse errors prevent the Heisenberg scaling $\xi_{\pm}^{-2}\sim N$ expected for the ideal dynamics, robust sequences can still achieve scalable spin squeezing up to $N=80$.
Finally, we further consider atomic arrays in optical tweezers realizing the Hamiltonian $H = \sum_{i<j}\frac{J}{|i-j|^{\alpha=3}}(X_iX_j+Y_iY_j)$ by encoding a qubit in two Rydberg levels of opposite parity for each atom~\cite{Browaeys2020rydberg}.
While such dynamics is equivalent to one-axis twisting for $\alpha=0$, recent works have theoretically and experimentally demonstrated that it can still attain scalable spin squeezing when $\alpha >0$~\cite{Comparin2022squeezing,Bornet2023squeezing,Block2024squeezing,koyluoglu2025squeezing}.
In this setting, two-mode spin squeezing can be generated by subensemble Hamiltonian engineering, which can be realized with static detuning masks~\cite{gong2023,Bornet2024detuning}.
By applying the sequences $\{X_{\pi/2}, Y_{-\pi/2}, I, Y_{\pi/2}, X_{-\pi/2}, I\}$ on $A$ and $\{X_{\pi/2}, Y_{-\pi/2}, X_{\pi}, Y_{-\pi/2}, X_{\pi/2}, I\}$ on $B$, we obtain the effective Hamiltonians
\begin{equation}
\begin{split}
    \tilde{H}_{AA,BB}&= \sum_{i<j}\frac{2J}{3|i-j|^{3}}(X_iX_j+Y_iY_j+Z_iZ_j)\\
    \tilde{H}_{AB}&= \sum_{i\in A, j\in B}\frac{2J}{3|i-j|^{3}}(X_iX_j+Y_iY_j)
\end{split}
\end{equation}
where the coupling constants are again optimal with respect to Thm.~\ref{th:necess}.
The inter-subensemble interactions now resemble 2A2S, while $\tilde{H}_{AA, BB}$ are Heisenberg Hamiltonians that stabilize the collective spin nature of each subensemble~\cite{koyluoglu2025squeezing}.
We numerically study the dynamics of two-mode spin squeezing generated by this dynamics in a rectangular lattice of $4\times 5$ atoms, where the sublattices $A$ and $B$ are alternated as in Fig.~\ref{fig:squeezing}(e), and report the results in Fig.~\ref{fig:squeezing}(f).
Remarkably, we observe two-mode spin squeezing up to $\xi_{\pm}^{-2}\sim4~\text{dB}$ under the exact dynamics (black line), comparable to the $\alpha=0$ case.
Using a robust pulse sequence with $J\tau_k = 3\times10^{-3}$, the same spin squeezing can be achieved even in presence of A- and F-errors with $\epsilon=\delta=0.03$, and PW-errors with $\tau_p=\tau_k/4$ for $\pi/2$-pulses,
where this choice of parameters is based on recent experiments~\cite{SM,Bornet2023squeezing}.
%


%
\prlsection{Discussion and outlook}
Our results unveil how a minimal control overhead, available in several experimental setups, is sufficient to engineer a large class of Hamiltonians.
Importantly, we also provide concrete prescriptions to realize these Hamiltonians in practice.
As an example, we numerically demonstrate two-mode spin squeezing generation in atomic systems addressing $n=2$ subensembles with robust pulse sequences.
Our approach may lead to practical realizations of higher-order Hamiltonian engineering~\cite{tyler2023higherorder, zhou2023higherorder, chen2026higherorder}, where in principle sufficiently long control sequences can attain universal quantum computation~\cite{hu2026optimalcontrol}.
Moreover, our robust pulse sequence design finds direct application in the implementation of digital quantum processors with global control~\cite{Cesa2023dual, Menta2025superconducting,white2026dual,Cesa2026dual}, where it is crucial to minimize physical error rates.

\let\oldaddcontentsline\addcontentsline
\renewcommand{\addcontentsline}[3]{}
\prlsection{Acknowledgements}                  
We acknowledge insightful discussions with Kenneth Wang, Ana Maria Rey, Bingtian Ye, Leon Zaporski, Giulia Semeghini.
We are grateful to Joonhee Choi and Hannes Bernien for valuable feedback on the manuscript.
We acknowledge financial support by
the ARO MURI grant W911NF25-1-0263, the NSF CAREER Award DMR-2237244, the Alfred P. Sloan Fellowship, and the Center for Ultracold Atoms (an NSF Physics Frontier Center, grant number PHY-2317134). 
WG is supported by the Hertz Foundation Fellowship.
MV acknowledges support from the European Union’s Horizon Europe research and innovation program under the Marie Sk{\l}odowska-Curie Action INDIANAQUSYS (Grant No. 101273040).

\bibliography{dual}

\onecolumngrid
\newpage
\twocolumngrid

\section{End matter}

Here we prove the necessary and the sufficient conditions for Hamiltonian engineering with subensemble control, stated in the main text as Thm.~\ref{th:necess} and Thm.~\ref{th:suff}. 
We begin each proof with the case of $n=2$ subensembles, and we generalize it afterwards.
Moreover, we prove how the symmetrization techniques in Fig.~\ref{fig:robust} provide robustness against amplitude (A), frequency (F), and pulse-width (PW) errors.

\subsection{Proof for the necessary conditions}

Let us first recall that local unitary transformations on a two-qudit Hamiltonian translate to orthogonal transformations on the interaction matrix~\cite{Choi17qudit}
\begin{equation}
\label{eq:orth}
\begin{split}
    (u_i \otimes v_j)^{\dagger} h_{ij} (u_i \otimes v_j) =\sum_{\mu\nu} [O_u^TgO_v]_{\mu\nu}\lambda_i^{\mu}\lambda_j^{\nu} \ ,    
\end{split}
\end{equation}
where $u_i, v_j$ are $SU(d)$ rotations, and $O_u, O_v$ are corresponding $SO(d^2-1)$ rotations.
Thereby, applying a pulse sequence in the $n=2$ subensemble case results in the effective interaction matrices
\begin{equation}
\label{eqend:convex_ortho}
\begin{split}
    \tilde{g}^{AA} &= \sum_k\frac{\tau_k}{\tau}O^{(k)T}_A g^{AA} O_A^{(k)} \ , \\ 
    \tilde{g}^{AB} &= \sum_k\frac{\tau_k}{\tau}O^{(k)T}_A g^{AB} O_B^{(k)} \ , 
\end{split}
\end{equation}
and similarly for $g^{BB}$, where $O_{A,B}^{(k)}$ are orthogonal matrices corresponding to the pulses $u^{(k)}_{A,B}$.
All these expressions can be grouped together in the compact form
\begin{equation}
\label{eqend:convex_ortho_AB}
    \tilde{G}^{AB} = \sum_k \frac{\tau_k}{\tau}O_{AB}^{(k)T}G^{AB}O_{AB}^{(k)}
\end{equation}
with
\begin{equation}
    G^{AB} = 
    \begin{pmatrix}
        {g}^{AA} & {g}^{AB} \\ ({g}^{AB})^T & {g}^{BB}
    \end{pmatrix}
    ~,~ O_{AB}^{(k)} = 
    \begin{pmatrix}
        O_{A}^{(k)} & 0\\
        0 & O_B^{(k)}
    \end{pmatrix}
    \ .
\end{equation}
When the matrices $g^{AA,BB}$ are symmetric, $G^{AB}$ is also symmetric. 
In this case, and given that $O_{AB}^{(k)}$ is an orthogonal matrix, Eq.~\eqref{eqend:convex_ortho} is equivalent to the majorization described in Thm.~\ref{th:necess}~\cite{Choi2003convex, bhatia2013matrix}, proving our result.
This result is analogous to the previously obtained necessary condition for global pulse sequences, as Eq.~\eqref{eqend:convex_ortho_AB} is equivalent to Eq.~\eqref{eqend:convex_ortho} for $g^{AA, BB}$~\cite{masanes2002global}, that need to be satisfied independently together with these conditions.
Moreover, Eq.~\eqref{eqend:convex_ortho_AB} cannot be a sufficient condition in general, as it neglects the block-diagonal structure of $O_{AB}^{(k)}$.
In fact, it is possible to construct examples where the necessary condition is satisfied, but Hamiltonian engineering would require nonzero off-block-diagonal terms in $O_{AB}^{(k)}$, akin to entangling operations between $A$ and $B$.
Let us now comment on the $n>2$ subensembles case.
Under the assumption that the native Hamiltonian is $2$-local, every term will only involve two subensembles.
Hence, one can write for every set of subensembles $S$
\begin{equation}
\label{eqend:convex_ortho_S}
    \tilde{G}^{S} = \sum_k \frac{\tau_k}{\tau}O_{S}^{(k)T}G^{S}O_{S}^{(k)}
\end{equation}
where $G^S$ has the intra-subensemble interaction matrices in $S$ in its diagonal blocks and interaction matrices between pairs of subensembles in $S$ in its off-diagonal blocks, while $O_S^{(k)}$ has the corresponding orthogonal rotations in its diagonal blocks.
For example, for $n=3$ we need to satisfy Eq.~\eqref{eqend:convex_ortho_S} for $S = ABC$, with
\begin{equation}
\label{eqend:blockG_S}
\begin{split}
    G^{ABC} &= 
    \begin{pmatrix}
        g^{AA} & g^{AB} & g^{AC} \\ (g^{AB})^T & g^{BB} & g^{BC} \\ (g^{AC})^T & (g^{BC})^T & g^{CC}
    \end{pmatrix}
\end{split}
\end{equation}
and $O_{ABC}^{(k)} = \text{diag}(O_A^{(k)}, O_B^{(k)}, O_C^{(k)})$, together with $S = AB, AC, BC$.
In the general case, this results in a family $2^n-n-1$ majorization conditions equivalent to Thm.~\ref{th:necess}.

\subsection{Proof for the sufficient condition}

As a first step, let us note that composing the pulse sequence $\{P^{(k')}, \tau'\}$ on top of the pulse sequence realizing the effective Hamiltonian $\tilde{H}$ results in
\begin{equation}
    \tilde{H}_{\text{comp}} = \sum_{k'}\frac{\tau'^{(k')}}{\tau'}U'^{(k')\dagger}\tilde{H}U'^{(k')} \ ,
\end{equation}
with $U'^{(k')} = \prod_{k''=1}^{k'} P'^{(k'')}$. 
By composing and alternating pulse sequences, proving Thm.~\ref{th:suff} for $n=2$ reduces to proving the existence of two pulse sequences that respectively achieve $\{g^{AA}, g^{AB}, g^{BB}\}\rightarrow \{0, \tilde{g}^{AB}, 0\}$ for a single $\tilde{g}^{AB}\neq 0$, and $\{g^{AA}, g^{AB}, g^{BB}\}\rightarrow \{\tilde{g}^{AA}, 0, \tilde{g}^{BB}\}$ with $\tilde{g}^{AA,BB}\neq 0$.
In fact, composing the first pulse sequence with another allows us to engineer arbitrary inter-subensemble interactions up to rescaling, since Hamiltonian engineering with local control is universal~\cite{Dodd02HamEng}.
Similarly, composing the second pulse sequence with another allows to engineer arbitrary intra-subensemble interactions given that $\text{tr}[g^{AA,BB}]=0$~\cite{masanes2002global} and that $A$ and $B$ are not interacting.
Let us now introduce a sequence $\{P^{*(k)}, \tau_k=\frac{1}{K}\}$ that generates a set of toggling-frame unitaries $\{u^{(k)}\}$ on individual qudits with a $2$-design property~\cite{Mele2024haar}.
Such pulse sequences have been constructed explicitly for arbitrary $d$~\cite{dur2000design, Choi17qudit}.
By applying $\{P^{*(k)}_A, \frac{1}{K}\}$ only on the subensemble $A$, we obtain $\tilde{h}_{BB} = h_{BB}$, and
\begin{equation}
\begin{split}
    \tilde{h}^{AB}_{ij} = \frac{1}{K}&\sum_k u_{A,i}^{(k)\dagger}h_{ij}^{AB}u_{A,i}^{(k)} = \frac{1}{d}I_i\otimes\text{tr}_i[h^{AB}_{ij}] =  0 \ , \\
    &\tilde{h}^{AA}_{ij} = \frac{\text{tr}[g^{AA}]}{d^2-1}\sum_{\mu}\lambda^{\mu}_i\otimes\lambda^{\mu}_j = 0 \ ,
\end{split}
\end{equation}
where $\text{tr}_i(\cdot)$ is a partial trace on the site $i$, and we have used the $2$-design property $\frac{1}{K}\sum_k f(u_k) = \int_{u\sim\text{Haar}} f(u)du$ up to second moment functions $f$, as well as the tracelessness of $g^{AA}$.
By alternating this pulse sequence with $\{P^{*(k)}_B, \frac{1}{K}\}$ applied on $B$, we fully decouple interactions in $AB$, while preserving the interactions in $AA$ and $BB$ half of the time.
Thus, we obtain $\{g^{AA}, g^{AB}, g^{BB}\}\rightarrow \{g^{AA}/2, 0, g^{BB}/2\}$, completing the first part of the proof.
Similarly, applying $\{P^{*(k)}_AP^{*(k)}_B, \frac{1}{K}\}$ results in $\tilde{h}^{AA,BB}_{ij}=0$, and
\begin{equation}
    \tilde{h}^{AB}_{ij} = \frac{\text{tr}[g^{AB}]}{d^2-1}\sum_{\mu}\lambda^{\mu}_i\otimes\lambda^{\mu}_j \ .
\end{equation}
When $\text{tr}[g^{AB}]\neq 0$, this sequence engineers $\{g^{AA}, g^{AB}, g^{BB}\}\rightarrow \{0, \tilde{g}^{AB}, 0\}$ with $\tilde{g}^{AB} \neq 0$, completing the proof.
Otherwise, it is always possible to compose $\{P^{*(k)}_AP^{*(k)}_B, \frac{1}{K}\}$ with any pulse sequence that yields $\{g^{AA}, g^{AB}, g^{BB}\}\rightarrow \{g'^{AA}, g'^{AB}, g'^{BB}\}$ such that $\text{tr}[g_{AB}']\neq 0$.
Given the universality of Hamiltonian engineering with local control, and the trace invariance of $g^{AA, BB}$, it is always possible to find such a pulse sequence.
For example, while $h^{AB}_{ij} = X_i X_j - Z_i Z_j$ has $\text{tr}[g^{AB}]=0$, it suffices to apply $X_{\pi}$ pulses on $A$ to get $h'^{AB}_{ij} = X_i X_j $ with $\text{tr}[g'^{AB}]=1$.
This proof technique can be extended to arbitrary $n$: it is sufficient to decouple interactions sequentially by composing the decoupling sequences above.
For instance, let us take $n=3$.
The sequence $\{P^{*(k)}_C, \frac{1}{K}\}$ decouples $AC$, $BC$ and $CC$ interactions, leaving $AA$, $AB$ and $BB$ interactions on.
At this point, the problem reduces to $n=2$ by composing pulse sequences; the same procedure can be repeated with $\{P^{*(k)}_A, \frac{1}{K}\}$ and $\{P^{*(k)}_B, \frac{1}{K}\}$.


\subsection{Robustness conditions}

Let us start from A-errors, which transform a pulse $O_\theta$ as $O_\theta \mapsto O_{\theta(1+\epsilon)} = O_{\theta\epsilon/2} O_\theta O_{\theta\epsilon/2}$, where $O_{\theta} = e^{-i\frac{\theta}{2} O}$ and $O$ is a Pauli operator acting on a chosen two-level subspace of a qudit.
Then, to leading order in $\epsilon$, the amplitude error can be regarded as a single-qudit error Hamiltonian evolution acting immediately before and after the ideal pulse. 
Each error Hamiltonians is then conjugated by the toggling frame unitaries $U^{(k)}=\bigotimes_ju_j^{(k)}$ induced by the pulse sequences, see Eq.~\eqref{eq:engH}.
When applying a pulse sequence $\{P^{(k)},\tau_k\}$ with $P^{(k)} = \bigotimes_j e^{-i\frac{\theta^{(k)}}{2}_jO^{(k)}_j} $, the total error contribution to the average Hamiltonian is 
\begin{align}
\label{eqend:ae}
    \delta \tilde H
    =
    \sum_{k,j}
    \frac{\epsilon\theta^{(k)}_j}{2\tau}
    \left(
    \hat{O}_j^{(k,k-1)}
    +
    \hat{O}_j^{(k,k)}
    \right) \ ,
\end{align}
where $\hat{O}_j^{(k,k')} = (U^{(k')})^\dagger  O_j^{(k)} U^{(k')}$.
The symmetrization prescription in Fig.~\ref{fig:robust} pairs every pulse $P^{(k)}$, which takes the toggling frame from $U^{(k-1)}$ to $U^{(k)}$, with an inverse pulse $O_{-\theta^{(k)}}^{(k)}$, which takes the toggling frame back from $U^{(k)}$ to $U^{(k-1)}$. Therefore, the reflected pulse sequence gives an error contribution equal to $-\delta\tilde H$, compensanting Eq.~\eqref{eqend:ae} exactly.

F-errors, on the other hand, add a $Z$ component to $X,Y$ pulses. 
They result in a similar error Hamiltonian as A errors, since $\exp(-i\frac{\theta}{2}( X+\delta  Z)) = \exp(i \frac{\delta}{2} Y) \exp(-i \frac{\theta}{2} X)\exp(-i \frac{\delta}{2} Y) + O(\delta^2)$, and similarly for $Y$ pulses. 
The resulting errors can be then compensated to leading order in $\delta$ using the symmetrization prescription in Fig.~\ref{fig:robust}.
Unlike the pulse-level errors discussed above, PW errors depend explicitly on the Hamiltonian. For simplicity, let us discuss the qubit case $d=2$, and extend to qudits in the Supplemental Materials~\cite{SM}. The error accumulated during the $k$-th pulse is
\begin{align}
\label{eqend:pwerr}
    \int_{0}^{\tau_p} U^{(k)\dagger}\left(P^{(k)}(t)^\dagger H P^{(k)} (t)\right)U^{(k)}  dt\ .
\end{align}
where $P(t) = \bigotimes_j O_{\theta_j(t)}$ is a partially applied pulse for a time $t\leq\tau_p$, and  $\theta(t) = \theta t/\tau_p$.
Since single-qudit operators in the Hamiltonian transform as
\begin{equation}
\label{eqend:pwrot}
\begin{aligned}
    O^{\dagger}_{\theta(t)}  \lambda^{\mu} O_{\theta(t)}
    = \cos(\theta(t))\lambda^{\mu}  +  \sin(\theta(t)) \frac{i[O,\lambda^{\mu}]}{2} \ ,
\end{aligned}
\end{equation}
each two-body term in the Hamiltonian results in squared terms $\sim \sin^2(\theta(t))$ and $\sim \cos^2(\theta(t))$, and cross-terms $\sim \sin(\theta(t))\cos(\theta(t))$. 
The contribution from the cross-terms can be corrected via the same symmetrization technique used for F-errors.
Specifically, if two qudits are subject to a same-angle $\theta$-pulse with width $\tau_p$, then during the application of the pulse they will be rotated of an angle $\theta(t) = \theta t/\tau_p$, and the reflected pulse rotates in the opposite direction for a time $(2\pi/\theta - 1)\tau_p$.
Pairing such terms in the error Hamiltonian then suffices to cancel them at the leading order in $\tau$, since
\begin{equation}
    \int_{(\frac{2\pi}{\theta}-1)\tau_p}^{\tau_p}\sin\left(\frac{\theta t}{\tau_p}\right)\cos\left(\frac{\theta t}{\tau_p}\right)dt =\int_0^{2\pi}\sin(2\theta')d\theta' = 0 \ .
\end{equation}
When two qudits are subject to pulses with different angles at the same time, it is then possible to separate the longer pulse in two: the first pulse is applied in parallel to the shorter pulse, and the other is applied independently.
The robustness condition above then applies straightforwardly.
Instead, the squared errors will always have nonzero time averages, but they can be mitigated by adjusting the pulse spacing to compensate for the additional induced Hamiltonian terms. 

\let\addcontentsline\oldaddcontentsline

\onecolumngrid
\newpage

\appendix
\setcounter{equation}{0}
\setcounter{figure}{0}
\setcounter{page}{1}
\renewcommand{\thefigure}{S\arabic{figure}}

\setcounter{secnumdepth}{2}

\begin{center}
    {\large \bf Supplemental Materials: Robust Hamiltonian engineering with subensemble control}
\end{center}

\tableofcontents

\section{Complexity of the Hamiltonian engineering problem}

In this section, we prove that determining the engineerability of target interactions is NP-hard in $d$ in the case of $n=2$ subensembles.
First, we prove the existence of a family of native Hamiltonians for which the Hamiltonian engineering problem with subensemble control is equivalent to the separability problem, hence proving its NP-hardness in $d$~\cite{Horodecki1996entanglement, Gurvits2003locc}.
Then, using standard tools from representation theory, we show that the Hamiltonian engineering problem with subensemble control, in the general case, is strictly harder than entanglement transformation problems with LOCC operations.
Interestingly, this points out that the necessary conditions in Thm.~\ref{th:suff}, while not being sufficient, are a state-of-the-art result for the interconvertibility of mixed states.

\subsection{Mapping to the separability problem and NP-hardness}

Let us consider the Hamiltonian engineering problem with subensemble control for $n=2$ subensembles.
In particular, given the native Hamiltonian
\begin{equation}
    H = \sum_{i<j\in A}J_{ij}h_{ij}^{AA} + \sum_{i\in A,j\in B}J_{ij}h_{ij}^{AB} + \sum_{i<j\in B}J_{ij}h_{ij}^{BB} \ ,
\end{equation}
we wish to determine whether it is possible to engineer 
\begin{equation}
    \tilde{H} = \sum_{i<j\in A}J_{ij}\tilde{h}_{ij}^{AA} + \sum_{i\in A,j\in B}J_{ij}\tilde{h}_{ij}^{AB} + \sum_{i<j\in B}J_{ij}\tilde{h}_{ij}^{BB} \ ,
\end{equation}
using pulse sequences with subensemble control.
As specified in the main text, we assume each two-body interaction $h_{ij}$ to be independent of $(i,j)$.
Since we only consider the leading order of the Floquet-Magnus expansion, we assume each two-body term $h_{ij}$ to strictly consist of two-body terms, \textit{i.e.} $h^{AA}_{ij} = \sum_{\mu\nu}g^{AA}_{\mu\nu}\lambda_{i}^{\mu}\otimes\lambda_{j}^{\nu}$ and similarly for $AB, BB$, where $g^{AA, AB, BB}_{\mu\nu}$ are interaction matrices and $\lambda_{i}^{\mu}$ with $\mu\in[1,d^2-1]$ are $d\times d$ generalized Gell-Mann matrices, \textit{i.e.} the qudit generalization of Pauli matrices~\cite{Choi17qudit}.
Applying the pulse sequence $\{P^{(k)}, \tau_k\}$, where $P^{(k)} = (\bigotimes_{i \in A} p_{A,i}^{(k)})\otimes( \bigotimes_{j \in B} p_{B,j}^{(k)})$, the resulting interactions are
\begin{equation}
    \tilde{h}_{ij}^{AA} = \sum_k \frac{\tau_k}{\tau}(u_{A,i}^{(k)}\otimes u_{A,j}^{(k)})^{\dagger} h_{ij}^{AA} (u_{A,i}^{(k)}\otimes u_{A,j}^{(k)}) \ , \ \tilde{h}_{ij}^{AB} = \sum_k \frac{\tau_k}{\tau}(u_{A,i}^{(k)}\otimes u_{B,j}^{(k)})^{\dagger} h_{ij}^{AB} (u_{A,i}^{(k)}\otimes u_{B,j}^{(k)}) \ ,
\end{equation}
where $u_{A,i}^{(k)} = \prod_{k'=1}^k p^{(k)}_{A,i}$ and similarly for $B$ and $\tilde{h}^{BB}_{ij}$.

We now prove the existence of a set of interactions $\{h_{ij}^{AA}, h_{ij}^{AB}, h_{ij}^{BB}\}$ such that determining if $\{\tilde{h}_{ij}^{AA}, \tilde{h}_{ij}^{AB}, \tilde{h}_{ij}^{BB}\}$ is engineerable is equivalent to determining whether a two-qudit density matrix $\tilde{\rho}_{AB}$ is separable.
The separability problem consists in determining whether a target two-qudit density matrix $\tilde{\rho}_{AB}$ can be expressed as 
\begin{equation}
\label{eqapp:separability}
    \tilde{\rho}_{AB} = \sum_k p_k (u^{(k)}_A\otimes u^{(k)}_B)^{\dagger}\rho_{AB} (u^{(k)}_A \otimes u^{(k)}_B) \ ,
\end{equation}
where $\rho_{AB} = \ket{\psi}\bra{\psi}_{A}\otimes \ket{\phi}\bra{\phi}_{B}$ is a fixed two-qudit pure pure product state, and $p_k>0$ such that $\sum_k p_k = 1$.
The form of Eq.~\eqref{eqapp:separability} resembles the one of a local Hamiltonian engineering problem: this mapping can be achieved by defining $p_k = \tau_k/\tau$ and
\begin{equation}
\label{eqapp:DM_mapping}
    {\rho}_{AB} = (\lambda_{s}I_{AB}-{h}_{AB})/\lambda_s d^2 \ , \ \tilde{\rho}_{AB} = (\lambda_{s}I_{AB}-\tilde{h}_{AB})/\lambda_s d^2 \ ,
\end{equation}
where $\lambda_s$ is the smallest eigenvalue of ${h}_{AB}$ and $I_{AB}$ is the identity, taking the convention $\text{tr}[{h}_{AB}]=0$.
However, two-qudit density matrices generally decompose as
\begin{equation}
    {\rho}_{AB} = \frac{1}{d^2}\left( I_{AB} + \sum_{\mu} (\text{tr}[\rho_{AB}\lambda_A^{\mu}]\lambda^{\mu}_A + \text{tr}[\rho_{AB}\lambda_B^{\mu}]\lambda^{\mu}_B) + \sum_{\mu\nu} \text{tr}[\rho_{AB}\lambda_A^{\mu}\lambda_B^{\nu}]\lambda^{\mu}_A\otimes\lambda_B^{\nu}\right) \ ,
\end{equation}
whereas in the local Hamiltonian engineering problem one typically disregards single-qudit terms~\cite{Bennett02HamEng}.
This distinction lies at the core of the tractability of the local Hamiltonian engineering problem: the absence of single-qudit terms excludes native Hamiltonians $h_{AB}$ that map to two-qudit pure product states.
Therefore, there is no instance of the local Hamiltonian engineering problem that maps to the separability problem.

We now show that a set of instances of the Hamiltonian engineering problem with control on $n=2$ subensembles admits a similar mapping to Eq.~\eqref{eqapp:separability}, but that includes every possible $\rho_{AB}$.
To achieve this, we note that intra-subensemble two-body Hamiltonians of the following form transform as single-qudit Hamiltonians under global pulse sequences
\begin{equation}
\label{eqapp:single_qudit_H}
    h_{ij}^{AA,BB} = \sum_{c}g_{c}'^{AA,BB}\sum_{ab}f_{abc}\lambda_i^{a}\otimes \lambda_{j}^{b} \ ,
\end{equation}
where $g_c'^{AA,BB}$ are arbitrary $d^2-1$-dimensional vectors, and $f_{abc}$ are the antisymmetric structure constants defined by the commutation relations of the generalized Gell-Mann matrices
\begin{equation}
    [\lambda^a, \lambda^b] = 2i\sum_c f_{abc}\lambda^c \ .
\end{equation}
In particular, given the single-qudit Hamiltonians $h'^{AA,BB}_i = \sum_c g'_c \lambda^c_i$, we have
\begin{equation}
    \sum_k \frac{\tau_k}{\tau} (u^{(k)}_i\otimes u^{(k)}_j)^{\dagger} h_{ij}(u^{(k)}_i\otimes u^{(k)}_j) = \sum_{c} \tilde{g}'_c\sum_{ab}f_{abc}\lambda_i^{a}\otimes \lambda_{j}^{b} \iff \sum_k \frac{\tau_k}{\tau} u^{(k)\dagger}_i h_{i}'u^{(k)}_i = \sum_{c} \tilde{g}'_c\lambda_i^{c} \ .
\end{equation}
This can be proven in two steps.
First, since
\begin{equation}
    [u^{\dagger}\lambda^a u, u^{\dagger}\lambda^b u] = 2i\sum_c f_{abc}u^{\dagger}\lambda^c u \ ,
\end{equation}
and knowing that $u^{\dagger}\lambda^a u = \sum_{a'}[O_u]_{aa'}\lambda^{a'}$, we obtain
\begin{equation}
    \sum_{a'b'}[O_u]_{aa'}[O_u]_{bb'}f_{a'b'c} = \sum_{c'}f_{abc'}[O_u]_{c'c} \ .
\end{equation}
From this, we prove that the operator $\sum_{ab}f_{abc}\lambda_i^{a}\otimes \lambda_{j}^{b}$ transforms under homogeneous local unitaries in the same way as $\lambda^c$
\begin{equation}
    (u_i\otimes u_j)^{\dagger}\left(\sum_{ab}f_{abc}\lambda_i^{a}\otimes \lambda_{j}^{b}\right)(u_i\otimes u_j) = \sum_{aa'bb'}f_{abc}[O_u]_{aa'}[O_u]_{bb'}\lambda_i^{a'}\otimes \lambda_{j}^{b'} = \sum_{c'}[O_u]_{cc'}\sum_{ab}f_{abc'}\lambda_i^{a}\otimes \lambda_{j}^{b} \ .
\end{equation}
This completes the proof.

As a consequence, for any two-qudit density matrix $\rho_{AB}$ in Eq.~\eqref{eqapp:separability} we can always find a triple $\{h^{AA}_{ij}, h^{AB}_{ij}, h^{BB}_{ij}\}$ that maps to it by choosing
\begin{equation}
\label{eqapp:H_from_rho}
    h^{AB}_{ij} = \sum_{\mu\nu} \text{tr}[\rho_{AB}\lambda^{\mu}_A\lambda^{\nu}_B]\lambda^{\mu}_i \otimes\lambda^{\nu}_i \ , h^{AA,BB}_{ij} = \sum_{abc}\text{tr}[\rho_{AB}\lambda_{A,B}^{c}]f_{abc}\lambda_i^a\otimes \lambda_j^b \ .
\end{equation}
Therefore, determining the engineerability of $\{h^{AA}, h^{AB}, h^{BB}\}\rightarrow\{\tilde{h}^{AA}_{ij}, \tilde{h}^{AB}_{ij}, \tilde{h}^{BB}_{ij}\}$ where $\tilde{h}^{AA,BB}_{ij}$ is mapped from an arbitrary two-qudit density matrix $\tilde{\rho}_{AB}$ using Eq.~\eqref{eqapp:H_from_rho}, and where $\{h^{AA}_{ij}, h^{AB}_{ij}, h^{BB}_{ij}\}$ is mapped from a pure product state, is as hard as determining if $\tilde{\rho}_{AB}$ is separable.
This implies that being able to solve the Hamiltonian engineering problem with control on $n=2$ subensembles requires being able to solve the separability problem, hence it is NP-hard in the local dimension $d$.

\subsection{Mapping to membership in $\text{CUT}_n^{\pm}$ and NP-hardness}

In the previous section, we showed that the Hamiltonian engineering problem with $n = 2$ subensembles is NP-hard in the qudit dimension $d$. Here, we show that in the qubit case $d = 2$, the Hamiltonian engineering problem is NP-hard in the number of subensembles $n$. We do so by mapping Hamiltonian engineering to  membership in $\text{CUT}_n^{\pm}$, which we explain below.

The membership problem in the cut polytope $\text{CUT}_n^{\pm}$ is defined as follows: given a rational symmetric matrix
\begin{align}
    Q = (q_{ab})^{n}_{a,b=1}, \,\,q_{aa} = 1,
\end{align}
we aim to decide whether
\begin{align}
    Q \in \text{CUT}_n^{\pm} \coloneq \text{conv}\{ss^T : s\in\{\pm 1\}^n\}.
\end{align}
This is related to the MAX-CUT optimization problem and is NP-hard~\cite{Pitowsky1991}. Moreover, as membership can be efficiently verified by checking a provided decomposition, this problem is NP-complete. Hamiltonian engineering similarly involves finding a decomposition of a target matrix into a convex combination of other matrices. Using this insight, we find an instance of the Hamiltonian engineering problem which reduces to this membership problem.

Specifically, we define an initial Hamiltonian and a class of target Hamiltonians parametrized by a rational symmetric matrix $Q$ such that determining engineerability of a target Hamiltonian is equivalent to deciding whether $Q \in \text{CUT}_n^{\pm}$. We consider a system with $n$ subensembles $S_a$, $a\in [1,n]$, each containing two qubits. Working in the interaction picture of Eq.~\eqref{eq:int_mat}, our initial Hamiltonian uses intra-subensemble interactions defined by
\begin{align}
    g^{S_aS_a} = \text{diag}(1, 2, 3) \equiv D
\end{align}
and inter-subensemble interactions
\begin{align}
    g^{S_aS_b} = \text{diag}{(1, 0, 0}) \equiv E \ .
\end{align}
The target Hamiltonian, given $Q$ an $n \times n$ rational, symmetric matrix with entries $q_{ab}$, is
\begin{align}
    \tilde{g}^{S_aS_a} = D, \,\, \tilde{g}^{S_aS_b} = q_{ab} E. 
\end{align}
In terms of Pauli operators, and for $S_a\neq S_b$,
\begin{align}
    h_{ij}^{S_aS_a} &= X_{i}X_{j} + 2Y_{i}Y_{j} + 3 Z_{i}Z_{j} \ , h_{ij}^{S_aS_b} = X_{i}X_{j},\label{eq:hinitialnp} \\
    \tilde{h}_{ij}^{S_aS_a} &= X_{i}X_{j} + 2Y_{i}Y_{j} + 3 Z_{i}Z_{j} \ , \tilde{h}_{ij}^{S_a S_b} = q_{ab} X_{i}X_{j}\label{eq:htargetnp}.
\end{align}

Now, we prove the following theorem,
\begin{theorem}
    Given $H$ defined by Eq.~\eqref{eq:hinitialnp} and $\tilde{H}(Q)$ defined by Eq.~\eqref{eq:htargetnp}, $\tilde{H}(Q)$ is engineerable from $H$ using subensemble control if and only if $Q \in \text{CUT}_n^{\pm}$.
\end{theorem}
\begin{proof}
    Recall that $\tilde{H}(Q)$ is engineerable from $H$ if there exists a pulse sequence $\{\tau_k,O_{S_a}^{(k)}\}_{a,k}$, where $O_{S_a}^{(k)} \in SO(3)$ is the adjoint orthogonal rotation associated with the $k$-th pulse on the $a$-th subensemble, such that
    \begin{align}\label{eqapp:hengnp}
        \tilde{g}^{S_aS_b} = \sum_k \frac{\tau_k}{\tau} (O_{S_a}^{(k)})^T g^{S_aS_b}O_{S_b}^{(k)}.
    \end{align}
We first prove the forward direction. Suppose $\tilde{H}(Q)$ is engineerable, so that Eq.~\eqref{eqapp:hengnp} implies
\begin{align}
    D  = \sum_k \frac{\tau_k}{\tau} (O_{S_a}^{(k)})^T D O_{S_a}^{(k)} \equiv \sum_k \frac{\tau_k}{\tau} D_{a, k}
\end{align}
for some pulse sequence $\{\tau_k,O_{S_a}^{(k)}\}_{a,k}$.

Now, $\Vert D_{a,k} \Vert_F = \Vert D\Vert_F$ as $D_{a,k}$ is an orthogonal rotation of $D$, where $\Vert \cdot \Vert_F$ indicates the Frobenius norm. Thus,
we can compute 
\begin{align}
    \sum_k \frac{\tau_k}{\tau} \Vert D_{a,k} - D\Vert_F^2 &= \sum_k \frac{\tau_k}{\tau}\Vert D_{a,k}\Vert_F^2 - 2\Tr \left(\sum_k\frac{\tau_k}{\tau} D^T_{a,k} D\right) +  \Vert D\Vert_F^2 \\
    &= \sum_k \frac{\tau_k}{\tau}\Vert D_{a,k}\Vert_F^2 - 2\Tr \left(\left(\sum_k\frac{\tau_k}{\tau} D_{a,k}\right)^T D\right) +  \Vert D\Vert_F^2 \\
    &= \Vert D\Vert_F^2 - 2\Vert D\Vert_F^2 +  \Vert D\Vert_F^2 = 0
\end{align}
where we have used $ \Vert D_{a,k} - D\Vert_F^2  = \Tr\left((D_{a,k} - D)^T (D_{a,k} - D)\right)$. As all terms on the left-hand side are non-negative, this implies that $\Vert D_{a,k} - D\Vert_F^2 = 0$ and thus 
\begin{align}
    D_{a,k} =(O_{S_a}^{(k)})^T D O_{S_a}^{(k)} = D
\end{align}
for all non-zero $\tau_k$. Thus, $O_{S_a}^{(k)}$ must consist of diagonal matrices with $\pm 1$ entries, mapping to $\pi-$rotations about the coordinate axes in $SU(2)$. So,
\begin{align}
    O_{S_a}^{(k)} = \text{diag}(\epsilon_{a,1}^{(k)}, \epsilon_{a,2}^{(k)}, \epsilon_{a,3}^{(k)})
\end{align} where $\epsilon_{a,i}^{(k)} \in \{\pm 1\}$. Let us denote $s_{a}^{(k)} = \epsilon_{a,1}^{(k)}$, so that $(O_{S_a}^{(k)})^T E O_{S_b}^{(k)} = s_{a}^{(k)} s_{b}^{(k)} E$. Eq.~\eqref{eqapp:hengnp} then implies that 
\begin{align}
    q_{ab}E = \sum_k \frac{\tau_k}{\tau}  s_{a}^{(k)} s_{b}^{(k)} E \implies q_{ab} =  \sum_k \frac{\tau_k}{\tau}  s_{a}^{(k)} s_{b}^{(k)} \implies Q =  \sum_k \frac{\tau_k}{\tau} \mathbf{s}^{(k)} (\mathbf{s}^{(k)})^T,
\end{align}
thus implying that $Q  \in \text{CUT}_n^{\pm}$.

Now, we prove the reverse direction. Suppose $Q  \in \text{CUT}_n^{\pm}$. Then, $Q =  \sum_k \frac{\tau_k}{\tau} \mathbf{s}^{(k)} (\mathbf{s}^{(k)})^T$ for some $n$-dimensional sign vectors $\mathbf{s}^{(k)}$. Each entry of $Q$ can therefore be described by $q_{ab} =  \sum_k \frac{\tau_k}{\tau}  s_{a}^{(k)} s_{b}^{(k)}$.  We can then assign rotations
\begin{align}
    O_{S_a}^{(k)} = \text{diag}(s_a^{(k)}, s_a^{(k)}, 1),
\end{align}
such that $ O_{S_a}^{(k)} \in SO(3)$ and $(O_{S_a}^{(k)})^T D O_{S_a}^{(k)} = D$. The pulse sequence $\{\tau_k,O_{S_a}^{(k)}\}_{a,k}$ realizes Eq.~\eqref{eqapp:hengnp}, so $\tilde{H}$ is engineerable from $H$.
\end{proof}

A polynomial-time algorithm determining engineerabilty for $d = 2$ and an arbitrary number of subensembles $n$ would therefore solve an NP-complete problem.

\subsection{Mapping to mixture of local unitaries}

In this section we show that the Hamiltonian engineering problem with control on $n=2$ subensembles and $d=2$ generally maps to a class of unsolved problems in entanglement theory.
To achieve this, we show that the mapping from Hamiltonian to density matrices introduced in the previous section maps arbitrary two-qubit Hamiltonians to two-qutrit states (\textit{i.e.} made of two three-level systems).
Then, we show that the set of transformations between two general density matrices is more constrained than local operations and classical communications (LOCC), breaking a standard assumption in the theory of entanglement transformations.

First, let us introduce the Fano representation of single-qudit operators~\cite{Fano1957,Fano1983,SP_2017}, which maps the transformation of spins$-s$ operators to the ones of states as
\begin{equation}
\label{eq:Fano}
    O_s = \sum_{l=0}^{2s}\sum_{m=-l}^l \text{tr}[O_d T_{l}^{(m)}] T_{l}^{(m)} \ ,
\end{equation}
where each operator $T_{l}^{(m)}$ transforms like a spin$-l$ state under $SU(2)$ transformations, and $\bra{s, m'} T_{l}^{(m)}\ket{s, m''} = C((s, s)\rightarrow l, (m', m'')\rightarrow m)$ are Clebsch-Gordan coefficients.
This yields the decomposition $O_2 = aI + \sum_{\mu} b_{\mu}\sigma^{\mu}$ for single-qubit operators, where the Pauli matrices $\sigma^{\mu}$ represent the spin$-1$ component, while the trivial spin$-0$ component is proportional to the identity operator $I = T_1^{(0)}$.
This representation allows one to compose operators subject to homogeneous $SU(2)$ transformations in the same way as angular momenta for states.
For example, the two-body Hamiltonian $h_{ij} = \sum_{\mu\nu}g_{\mu\nu}\sigma^{\mu}_{i}\sigma^{\nu}_{j}$ admits a Fano representation equivalent to the composition of two spin$-1$ states, \textit{i.e.}
\begin{equation}
    h = \text{tr}[hT_0^{(0)}]T_0^{0} + \sum_{m=-1}^1\text{tr}[hT_1^{(m)}]T_1^{m} + \sum_{m=-2}^2\text{tr}[hT_2^{(m)}]T_2^{m} \ ,
\end{equation}
where now $T_0^{(0)}\sim X_iX_j + Y_iY_j + Z_iZ_j$ is the Heisenberg-like component, $T_1^{(m)}$ are the antisymmetric interactions previously obtained in Eq.~\eqref{eqapp:single_qudit_H}, and $T_2^{(0)}\sim T_0^{(0)}-3Z_iZ_j$ are symmetric dipolar interactions. 
This decomposition is the starting point for Hamiltonian engineering problems with global control~\cite{masanes2002global}, and is equivalent to Eq.~\eqref{eq:Fano} with $s=1$, implying that any two-qubit operator $h$ under global $SU(2)$ rotations transform like a single-qutrit operator $h'$.
Therefore, using Fano representation and similar tools to the previous section, it is now possible to map any Hamiltonian engineering problem with $n=2, d=2$ subensemble control to Eq.~\eqref{eqapp:separability} under the mapping
\begin{equation}
    {\rho}_{AB} = \frac{\lambda_{s}I_{AB}-K_{AB}}{9\lambda_s}  \ , \ \tilde{\rho}_{AB} = \frac{\lambda_{s}I_{AB}-\tilde{K}_{AB}}{9\lambda_s} \ ,
\end{equation}
where $\lambda_s$ is the smallest eigenvalue of $K_{AB}$, defined as
\begin{equation}
    K_{AB} = h'_{AA} \otimes I_B + I_A \otimes h'_{BB} + h_{AB}' \ ,
\end{equation}
and $h'_{AA,BB}$ are the single-qutrit mapping of $h_{AA,BB}$, while $h'_{AB}$ is obtained by promoting the Pauli matrices in $h_{AB}$ to the corresponding spin$-1$ operators.
An equivalent definition applies to $\tilde{K}_{AB}$.

Let us comment on this result.
First, the state transformation problem of Eq.~\eqref{eqapp:separability} is here a two-qutrit problem, while the unitaries $u^{(k)}$ are restricted to $SU(2)$.
While constrained, this problem is strictly more general than the separability problem. 
Hence, the solution of the separability problem for $d=2$ does not apply here~\cite{Horodecki1996entanglement}.
Pairs of states $\{\rho, \tilde{\rho}\}$ related by Eq.~\eqref{eqapp:separability} for arbitrary $\rho_{AB}$ are said to be related by mixture of local unitaries.
Mixture of local unitaries are fundamentally harder to characterize than LOCC operations, as the equivalence between density matrices and statistical ensembels of pure states does not hold for the former.
Hence, known results for LOCC transformations~\cite{Nielsen1999entanglement,li2011entanglement,lin2024entanglement} do not apply, requiring drastically different techniques.
In particular, we expect a characterization based on local unitary invariants~\cite{Makhlin2002LU}.
Interestingly, the local unitary invariants identified in the two-qudit case of Ref.~\cite{zhou2024LU} are related to the interaction matrices $g^{AA,AB,BB}$.
Therefore, we expect our necessary condition Thm.~\ref{th:necess} to be a relevant result in this context.

\section{Details on numerical pulse sequence design}

\subsection{Hamiltonian engineering as an optimization problem}

Here, we formulate subensemble Hamiltonian engineering as an optimization problem compatible with numerical solvers. We continue to work in the interaction-matrix representation.

Suppose the initial Hamiltonian is specified by interaction matrices $\{g^{\mathrm{intra},a}\}_a$ and $\{g^{\mathrm{inter},b}\}_b$, where the former denote intra-subensemble interactions, such as $AA$ or $BB$, and the latter denote inter-subensemble interactions, such as $AB$. Similarly, let the target Hamiltonian be specified by $\{\tilde g^{\mathrm{intra},a}\}_a$ and $\{\tilde g^{\mathrm{inter},b}\}_b$. For each intra-subensemble target interaction, we decompose 
\begin{align}
    \tilde g^{\mathrm{intra},a}
    =
    \alpha_a I+\beta_a \bar g^{\mathrm{intra},a},
\end{align}
where
\begin{align}
    \alpha_a
    =
    \frac{\Tr(\tilde g^{\mathrm{intra},a})}{d^2-1},
    \qquad
    \beta_a
    =
    \left\|\tilde g^{\mathrm{intra},a}-\alpha_a I\right\|_F,
    \qquad
    \bar g^{\mathrm{intra},a}
    =
    \frac{\tilde g^{\mathrm{intra},a}-\alpha_a I}
    {\left\|\tilde g^{\mathrm{intra},a}-\alpha_a I\right\|_F},
\end{align}
and $\left \lVert A \right \rVert_F:=\left(
    \sum_{i=1}^m\sum_{j=1}^n |A_{ij}|^2
    \right)^{1/2}$ is the Frobenius norm of an operator $A$. 
Here, $I$ denotes the Heisenberg interaction matrix. Similarly, for each inter-subensemble target interaction, we write
\begin{align}
    \tilde g^{\mathrm{inter},b}
    =
    \gamma_b \bar g^{\mathrm{inter},b},
    \qquad
    \gamma_b
    =
    \left\|\tilde g^{\mathrm{inter},b}\right\|_F,
    \qquad
    \bar g^{\mathrm{inter},b}
    =
    \frac{\tilde g^{\mathrm{inter},b}}
    {\left\|\tilde g^{\mathrm{inter},b}\right\|_F}.
\end{align}

Given a pulse sequence $\{\tau_k,O_{A_a}^{(k)}\}_{a,k}$, where $O_{A_a}^{(k)}$ is the adjoint orthogonal rotation associated with the $k$-th pulse on the $a$-th subensemble, define the effective interaction matrices engineered by this sequence as
\begin{align}
    g_{\mathrm{eff}}^{\mathrm{intra},a}
    &:=
    \sum_k \frac{\tau_k}{\tau}
    (O_{A_a}^{(k)})^T g^{\mathrm{intra},a} O_{A_a}^{(k)},
    \\
    g_{\mathrm{eff}}^{\mathrm{inter},ab}
    &:=
    \sum_k \frac{\tau_k}{\tau}
    (O_{A_a}^{(k)})^T g^{\mathrm{inter},ab} O_{A_b}^{(k)}.
\end{align}
Hamiltonian engineering can then be framed as the following optimization problem:
\begin{align}\label{eq:numopt}
    \max_{\{\tau_k\},\{O_a^{(k)}\},\{v_a\},\{w_{ab}\}}
    \quad
    &\sum_a v_a+\sum_{a<b}w_{ab}
    \\
    \text{s.t.}\quad
    &
    \sum_k \frac{\tau_k}{\tau}
    (O_{A_a}^{(k)})^T g^{\mathrm{intra},a} O_{A_a}^{(k)}
    =
    \alpha_a I+v_a \bar g^{\mathrm{intra},a},
    \qquad \forall a,
    \\
    &
    \sum_k \frac{\tau_k}{\tau}
    (O_{A_a}^{(k)})^T g^{\mathrm{inter},ab} O_{A_b}^{(k)}
    =
    w_{ab}\bar g^{\mathrm{inter},ab},
    \qquad \forall a<b,
    \\
    &
    \sum_k \frac{\tau_k}{\tau}=1,
    \qquad
    \frac{\tau_k}{\tau}\geq 0.
\end{align} 
The optimal values $v_a^{\max}$ and $w_{ab}^{\max}$ give the largest achievable interaction strengths along the prescribed target directions. Since traceless interaction matrices can always be canceled by an appropriate decoupling sequence, any target satisfying
\begin{align}
    \beta_a\leq v_a^{\max},
    \qquad
    \gamma_{ab}\leq w_{ab}^{\max}
\end{align}
is engineerable by concatenating the optimal sequence with a decoupling sequence. Thus, this optimization captures the reachable region for the target interaction matrices.

To implement this numerically, we impose the constraint $\sum_k \tau_k = \tau$, as well as alignment constraints requiring the engineered interaction matrices to lie along the desired target directions. Define
\begin{align}
    \mathbf{g}^{\mathrm{intra},a}_{\mathrm{eff}}
    &:=
    \operatorname{vec}
    \left(
    g_{\mathrm{eff}}^{\mathrm{intra},a}-\alpha_a I
    \right),
    &
    \bar{\mathbf{g}}^{\mathrm{intra},a}
    &:=
    \operatorname{vec}
    \left(
    \bar g^{\mathrm{intra},a}
    \right),
    \\
    \mathbf{g}^{\mathrm{inter},ab}_{\mathrm{eff}}
    &:=
    \operatorname{vec}
    \left(
    g_{\mathrm{eff}}^{\mathrm{inter},ab}
    \right),
    &
    \bar{\mathbf{g}}^{\mathrm{inter},ab}
    &:=
    \operatorname{vec}
    \left(
    \bar g^{\mathrm{inter},ab}
    \right),
\end{align}
where $\operatorname{vec}(g)$ denotes vectorization of a matrix $g$. Since the target directions are normalized in Frobenius norm, their vectorized forms have unit Euclidean norm. We then impose tracelessness as
\begin{align}
    (P_\perp)_a\mathbf{g}^{\mathrm{intra},a}_{\mathrm{eff}}
    &=0,
    \\
    (P_\perp)_{ab}\mathbf{g}^{\mathrm{inter},ab}_{\mathrm{eff}}
    &=0,
\end{align}
where
\begin{align}
    (P_\perp)_a
    =
    I
    -
    \bar{\mathbf{g}}^{\mathrm{intra},a}
    \left(\bar{\mathbf{g}}^{\mathrm{intra},a}\right)^T,
\end{align}
and similarly
\begin{align}
    (P_\perp)_{ab}
    =
    I
    -
    \bar{\mathbf{g}}^{\mathrm{inter},ab}
    \left(\bar{\mathbf{g}}^{\mathrm{inter},ab}\right)^T.
\end{align}
These constraints can be incorporated into the numerical objective through squared penalty terms, for instance
\begin{align}
    \mathcal{L}
    &=
    -\sum_a
    \mathbf{g}^{\mathrm{intra},a}_{\mathrm{eff}}
    \cdot
    \bar{\mathbf{g}}^{\mathrm{intra},a}
    -
    \sum_{a<b}
    \mathbf{g}^{\mathrm{inter},ab}_{\mathrm{eff}}
    \cdot
    \bar{\mathbf{g}}^{\mathrm{inter},ab}
    \nonumber \\
    &\quad
    +
    \eta_0
    \left(
    \sum_k \frac{\tau_k}{\tau}-1
    \right)^2
    +
    \eta_1
    \sum_a
    \left\|
    (P_\perp)_a
    \mathbf{g}^{\mathrm{intra},a}_{\mathrm{eff}}
    \right\|^2
    +
    \eta_2
    \sum_{a<b}
    \left\|
    (P_\perp)_{ab}
    \mathbf{g}^{\mathrm{inter},ab}_{\mathrm{eff}}
    \right\|^2,
\end{align}
where $\eta_0,\eta_1,\eta_2>0$ are penalty parameters. Equivalently, maximizing the first two terms gives the largest achievable coefficients $v_i$ and $w_{ij}$, while the penalty terms enforce normalization and alignment with the desired target directions.

The same engineering condition can be written compactly by collecting all interaction matrices into a single block matrix
\begin{align}
    G
    =
    \begin{pmatrix}
        g^{\mathrm{intra},1} & g^{\mathrm{inter},12} & \cdots & g^{\mathrm{inter},1n} \\
        g^{\mathrm{inter},21} & g^{\mathrm{intra},2} & \cdots & g^{\mathrm{inter},2n} \\
        \vdots & \vdots & \ddots & \vdots \\
        g^{\mathrm{inter},n1} & g^{\mathrm{inter},n2} & \cdots & g^{\mathrm{intra},n}
    \end{pmatrix}.
\end{align}
For each pulse $k$, define the block-diagonal orthogonal matrix
\begin{align}
    O^{(k)}
    =
    \operatorname{diag}
    \left(
    O_{A_1}^{(k)},O_{A_2}^{(k)},\dots,O_{A_n}^{(k)}
    \right) \ .
\end{align}
Then the full engineering equation can be written as
\begin{align}
    \sum_k \frac{\tau_k}{\tau}
    (O^{(k)})^T G O^{(k)}
    =
     \tilde G.
\end{align}
This expresses the engineered block interaction matrix as a convex combination of points in the orbit of $G$ under local adjoint rotations.

As $G$ is an $n(d^2-1)\times n(d^2-1)$ block interaction matrix, it has at most
\begin{align}
    D=n^2(d^2-1)^2
\end{align}
real entries. Therefore, by Carathéodory's theorem, any point in the convex hull of the orbit of $G$ can be expressed as a convex combination of at most
\begin{align}
    D+1=n^2(d^2-1)^2+1
\end{align}
points. Thus, it suffices to search over pulse sequences of length at most
\begin{align}
    L\leq n^2(d^2-1)^2+1.
\end{align}

Each local orthogonal rotation $O_{A_a}^{(k)}$ must arise from the adjoint action of some $SU(d)$ unitary. We parametrize it by
\begin{align}
    O_{A_a}(\omega)_{\nu',\nu}
    =
    \frac{1}{2}
    \Tr\!\left[
    \lambda^\nu u_{A_a}^\dagger(\omega)\lambda^{\nu'}u_{A_a}(\omega)
    \right] \ ,
\end{align}
where
\begin{align}
    u_{A_a}(\theta)
    =
    \exp\left(
    -i\sum_{\nu=1}^{d^2-1}\theta_{a,\nu}\lambda^\nu
    \right) \ .
\end{align}

Thus, for a sequence of length $L$, the optimization involves the pulse weights $\tau_k/\tau$ together with $n(d^2-1)$ rotation parameters per pulse. Using the Carathéodory bound, the total number of parameters is
\begin{align}
    \left(n^2(d^2-1)^2+1\right)
    \left(n(d^2-1)+1\right) = O(n^3d^6) \ .
\end{align}
For small $n$ and $d$, this finite-dimensional nonconvex optimization problem can be approached numerically using global optimization methods such as differential evolution.

\subsection{Completeness of pulse sets}

While Hamiltonian engineering can be approached via direct numerical optimization, the problem simplifies substantially when the allowed pulses on each subensemble are restricted to a finite set. In that case, the engineering problem can be formulated as a linear program (LP)~\cite{Choi17qudit}.

It is therefore useful to identify finite sets of unitaries that are sufficiently expressive. We call a finite set of toggling-frame unitaries $\mathcal U\subset SU(d)$ \textit{complete} if, for every interaction matrix $g$, the linear span of its adjoint orbit under $\mathcal U$ agrees with the linear span of its adjoint orbit under the full group $SU(d)$:
\begin{align}
\operatorname{span}
\left\{
O_{U_i}^T g^{ij} O_{U_j}: U_i,U_j\in \mathcal U
\right\}
=
\operatorname{span}
\left\{
O_{U_i}^T g^{ij} O_{U_j}: U_i,U_j\in SU(d)
\right\}.
\end{align}
where $O_U$ denotes the adjoint orthogonal representation of $U$.

At first sight, a linear-span condition is weaker than what is directly required for Hamiltonian engineering, since Hamiltonian engineering is defined by convex combinations of toggling-frame Hamiltonians. However, if both positive and negative evolution under the initial Hamiltonian are available, then linear combinations can be implemented by convex combinations. Indeed, any signed linear combination
\begin{align}
    H_{\mathrm{eff}}
    =
    \sum_\ell c_\ell H_\ell,
    \qquad
    c_\ell \in \mathbb{R},
\end{align}
can be rewritten as
\begin{align}
    H_{\mathrm{eff}}
    =
    \sum_{\ell:c_\ell>0} |c_\ell| H_\ell
    +
    \sum_{\ell:c_\ell<0} |c_\ell|(-H_\ell).
\end{align}
After normalizing the total evolution time, this becomes a convex combination of Hamiltonians drawn from the enlarged set containing both $H_\ell$ and $-H_\ell$. Thus, when both signs are available, equality of linear spans is sufficient to preserve the reachable set up to an overall scaling.

The same reasoning applies when the relevant part of the Hamiltonian is traceless. Since traceless interaction matrices $g$ can be canceled by Hamiltonian engineering, one can simulate their negative $-g$ by using a decoupling sequence. Concretely, 
suppose a sequence $\{U^{(k)}\}_k$ cancels a Hamiltonian $H$,
\begin{align}
    \sum_k p_k (U^{(k)})^\dagger H U^{(k)} = 0.
\end{align}
The sequence can be transformed to $\{(U^{(0)})^\dagger U^{(k)}\}_k$, such that
\begin{align}
    p_0 H +  \sum_{k \geq 1} p_k (U^{(k)})^\dagger U^{(0)} H (U^{(0)})^\dagger U^{(k)} = 0,
\end{align}
implying
\begin{align}
    \sum_{k \geq 1} p_k (U^{(k)})^\dagger U^{(0)} H (U^{(0)})^\dagger U^{(k)} = -p_0 H.
\end{align}
Thus, as long as $H$ can be decoupled, then the negative direction $-H$ is also engineerable up to an overall positive rescaling. Consequently, for traceless components, a linear-span condition is sufficient even though the original engineering problem is convex.

This reduces the problem to the optimization considered above Eq.~\eqref{eq:numopt}: the magnitude of the traceful Heisenberg component is fixed, while the goal is to maximize the projection of the engineered traceless component onto the desired traceless target direction. In particular, if the optimal values $v_a^{\max}$ and $w_{ab}^{\max}$ in Eq.~\eqref{eq:numopt} are positive, then a LP that optimizes on a set of pulses that is complete in the sense defined above can find a feasible solution with 
\begin{align}
    0<v_a\leq v_a^{\max},
    \qquad
    0<w_{ab}\leq w_{ab}^{\max}.
\end{align}

We now show that if the set of toggling-frame unitaries forms a unitary $4$-design~\cite{Mele2024haar}, then it is complete in the above sense. First, let us introduce frame operators.

\begin{lemma}\label{lem:frame_span}
    Let $\{v_\alpha\}_{\alpha\in G}$ be a finite set of vectors in a finite-dimensional Hilbert space $\mathcal H$, and define the frame operator
    \begin{align}
        F
        =
        \sum_{\alpha\in G}
        \ket{v_\alpha}\bra{v_\alpha}.
    \end{align}
    Then
    \begin{align}
        \operatorname{supp}F
        =
        \operatorname{span}\{v_\alpha:\alpha\in G\},
    \end{align}
where the support $\operatorname{supp}(F)$ of a positive semidefinite operator $F$ is the subspace on which it acts nontrivially.
\end{lemma}

\begin{proof}
    For any $x\in\mathcal H$,
    \begin{align}
        \braket{x,Fx}
        =
        \sum_{\alpha\in G}
        |\braket{v_\alpha,x}|^2.
    \end{align}
    Hence
    \begin{align}
        x\in \ker F
        \iff
        \braket{x,Fx}=0
        \iff
        \braket{v_\alpha,x}=0
        \quad \forall \alpha\in G.
    \end{align}
    Therefore,
    \begin{align}
        \ker F
        =
        \operatorname{span}\{v_\alpha:\alpha\in G\}^{\perp}.
    \end{align}
    Since $F$ is positive semidefinite, its support is the orthogonal complement of its kernel. Thus,
    \begin{align}
        \operatorname{supp}F
        =
        (\ker F)^\perp
        =
        \operatorname{span}\{v_\alpha:\alpha\in G\}.
    \end{align}
\end{proof}

\begin{theorem}\label{thm:four_design_complete}
    Let $\mathcal U\subset SU(d)$ be a unitary $4$-design. For each subensemble $A_a$, restrict the allowed local toggling-frame unitaries to $u_a\in \mathcal U$. Then $\mathcal U$ is complete for subensemble Hamiltonian engineering in the following sense.

    For every intra-subensemble interaction Hamiltonian $H_{A_aA_a}$,
    \begin{align}
        \operatorname{span}
        \left\{
        (u\otimes u)^\dagger H_{A_aA_a}(u\otimes u)
        :
        u\in \mathcal U
        \right\}
        =
        \operatorname{span}
        \left\{
        (u\otimes u)^\dagger H_{A_aA_a}(u\otimes u)
        :
        u\in SU(d)
        \right\}.
    \end{align}
    Moreover, for every inter-subensemble interaction Hamiltonian $H_{A_a A_b }$,
    \begin{align}
        \operatorname{span}
        \left\{
        (u\otimes v)^\dagger H_{A_aA_b}(u\otimes v)
        :
        u,v\in \mathcal U
        \right\}
        =
        \operatorname{span}
        \left\{
        (u\otimes v)^\dagger H_{A_aA_b}(u\otimes v)
        :
        u,v\in SU(d)
        \right\}.
    \end{align}
    Consequently, the set of product toggling-frame unitaries
    \begin{align}
        \mathcal U^{\times n}
        =
        \left\{
        \bigotimes_{a=1}^n u_a
        :
        u_a\in \mathcal U
        \right\}
    \end{align}
    is complete.
\end{theorem}

\begin{proof}
    We first prove the intra-subensemble statement. After vectorization, conjugation by $u\otimes u$ acts linearly as
    \begin{align}
        \mathcal E_{\mathrm{intra}}(u)
        =
        u^{\otimes 2}\otimes (u^*)^{\otimes 2}.
    \end{align}
    Therefore,
    \begin{align}
        \operatorname{vec}
        \left(
        (u\otimes u)^\dagger H_{A_aA_a}(u\otimes u)
        \right)
        =
        \mathcal E_{\mathrm{intra}}(u)^\dagger
        \operatorname{vec}(H_{A_aA_a}),
    \end{align}
    up to the convention chosen for vectorization. Thus, it suffices to show that
    \begin{align}
        \operatorname{span}
        \left\{
        \mathcal E_{\mathrm{intra}}(u):u\in \mathcal U
        \right\}
        =
        \operatorname{span}
        \left\{
       \mathcal E_{\mathrm{intra}}(u):u\in SU(d)
        \right\}.
    \end{align}

    Regard each matrix $\rho_{\mathrm{intra}}(u)$ as a vector in the Hilbert--Schmidt space of linear operators. Define
    \begin{align}
        F_\mathcal U^{\mathrm{intra}}
        &=
        \frac{1}{|\mathcal U|}
        \sum_{u\in \mathcal U}
        \ket{\mathcal E_{\mathrm{intra}}(u)}
        \bra{\mathcal E_{\mathrm{intra}}(u)},
        \\
        F_{\mathrm{Haar}}^{\mathrm{intra}}
        &=
        \int_{SU(d)}
        \ket{\mathcal E_{\mathrm{intra}}(u)}
        \bra{\mathcal E_{\mathrm{intra}}(u)}
        \,du.
    \end{align}
    The entries of
    \begin{align}
        \ket{\mathcal E_{\mathrm{intra}}(u)}
        \bra{\mathcal E_{\mathrm{intra}}(u)}
    \end{align}
    are polynomials of degree $4$ in the matrix entries of $u$ and degree $4$ in the matrix entries of $u^*$. Since $\mathcal{U}$ is a unitary $4$-design,
    \begin{align}
        F_\mathcal U^{\mathrm{intra}}
        =
        F_{\mathrm{Haar}}^{\mathrm{intra}}.
    \end{align}
    By Lemma~\ref{lem:frame_span},
    \begin{align}
        \operatorname{span}
        \left\{
        \mathcal E_{\mathrm{intra}}(u):u\in \mathcal U
        \right\}
        =
        \operatorname{span}
        \left\{
        \mathcal E_{\mathrm{intra}}(u):u\in SU(d)
        \right\}.
    \end{align}
    Applying both sides to $\operatorname{vec}(H_{A_aA_a})$ gives
    \begin{align}
        \operatorname{span}
        \left\{
        (u\otimes u)^\dagger H_{A_aA_a}(u\otimes u)
        :
        u\in \mathcal U
        \right\}
        =
        \operatorname{span}
        \left\{
        (u\otimes u)^\dagger H_{A_aA_a}(u\otimes u)
        :
        u\in SU(d)
        \right\}.
    \end{align}

    We now prove the inter-subensemble statement. After vectorization, conjugation by $u\otimes v$ acts linearly as
    \begin{align}
        \mathcal E_{\mathrm{inter}}(u,v)
        =
        (u\otimes u^*)\boxtimes (v\otimes v^*),
    \end{align}
    where $\boxtimes$ denotes the induced tensor-product action on the vectorized interaction matrix. Hence it suffices to show that
    \begin{align}
        \operatorname{span}
        \left\{
        \mathcal E_{\mathrm{inter}}(u,v):u,v\in \mathcal U
        \right\}
        =
        \operatorname{span}
        \left\{
        \mathcal E_{\mathrm{inter}}(u,v):u,v\in SU(d)
        \right\}.
    \end{align}

    Define
    \begin{align}
        F_\mathcal U^{\mathrm{inter}}
        &=
        \frac{1}{|\mathcal U|^2}
        \sum_{u,v\in \mathcal U}
        \ket{\mathcal E_{\mathrm{inter}}(u,v)}
        \bra{\mathcal E_{\mathrm{inter}}(u,v)},
        \\
        F_{\mathrm{Haar}}^{\mathrm{inter}}
        &=
        \int_{SU(d)}
        \int_{SU(d)}
        \ket{\mathcal E_{\mathrm{inter}}(u,v)}
        \bra{\mathcal E_{\mathrm{inter}}(u,v)}
        \,du\,dv.
    \end{align}
    The entries of
    \begin{align}
        \ket{\mathcal E_{\mathrm{inter}}(u,v)}
        \bra{\mathcal E_{\mathrm{inter}}(u,v)}
    \end{align}
    are polynomials of degree $2$ in the entries of $u$ and degree $2$ in the entries of $u^*$, and similarly degree $2$ in the entries of $v$ and degree $2$ in the entries of $v^*$. Thus a unitary $2$-design suffices to imply
    \begin{align}
        F_\mathcal U^{\mathrm{inter}}
        =
        F_{\mathrm{Haar}}^{\mathrm{inter}}.
    \end{align}
    Since every unitary $4$-design is also a unitary $2$-design, the equality holds. By Lemma~\ref{lem:frame_span},
    \begin{align}
        \operatorname{span}
        \left\{
        \mathcal E_{\mathrm{inter}}(u,v):u,v\in \mathcal U
        \right\}
        =
        \operatorname{span}
        \left\{
        \mathcal E_{\mathrm{inter}}(u,v):u,v\in SU(d)
        \right\}.
    \end{align}
    Applying both sides to $\operatorname{vec}(H_{A_aA_b})$ gives
    \begin{align}
        \operatorname{span}
        \left\{
        (u\otimes v)^\dagger H_{A_aA_b}(u\otimes v)
        :
        u,v\in \mathcal U
        \right\}
        =
        \operatorname{span}
        \left\{
        (u\otimes v)^\dagger H_{A_aA_b}(u\otimes v)
        :
        u,v\in SU(d)
        \right\}.
    \end{align}
    Therefore $\mathcal U$ is complete for both intra-subensemble and inter-subensemble interactions, and $\mathcal U^{\times n}$ is complete for subensemble Hamiltonian engineering.
\end{proof}

Our Theorem~\ref{thm:four_design_complete} formally proves that  a set of toggling-frame unitaries that form a unitary $4$-design is complete, and thus sufficiently expressive to use in a LP. Indeed, the isocahedral pulses investigated in Ref.~\cite{Ben2020iso} form a $4$-design in SU(2).

\section{Robustness against pulse-width errors}
Here, we comment on how robustness to PW-errors follows for qudits.
For general $d$, Eq.~\eqref{eqend:pwrot} becomes
\begin{equation}
\label{eqapp:pwrot_d}
\begin{aligned}
    O^{\dagger}_{\theta(t)}  \lambda^{\mu} O_{\theta(t)}
    = \cos(\alpha\theta(t))\lambda^{\mu}  +  \sin(\alpha\theta(t)) \frac{i[O,\lambda^{\mu}]}{2\alpha} \ ,
\end{aligned}
\end{equation}
when $\lambda^{\mu}$ is an off-diagonal Gell-Mann operator. Here, $\alpha = 1$ when $O$ and $\lambda^{\mu}$ act on the same two-level subspace, and $\alpha = \frac{1}{2}$ when the two operators overlap on one level. Otherwise, for $\lambda^{\mu}$ and $O$ with disjoint support, $\lambda^{\mu}$ is unchanged, \emph{i.e.} $\theta(t) = 0$. 

For diagonal $\lambda^{\mu} = \operatorname{diag}(h_1,\ldots,h_d)$, rotation by $O$ takes a slightly different form. Without loss of generality, suppose $O$ acts on the subspace spanned by the $m$ and $n$-th levels. We define
\begin{align}
    Z_{mn} = \ket m\bra m-\ket n\bra n, \quad
    \bar\lambda^{\mu}= \lambda^{\mu}-cZ_{mn}, \quad c = \frac{h_m-h_n}{2},
\end{align}
so that $Z_{mn}$ is a Pauli Z-type operator acting on the $m$ and $n$-th levels and $\lambda^{\mu}_{i, \mathrm{fixed}}$ commutes with $O$. Then,
\begin{align}
    O^{\dagger}_{\theta(t)}  \lambda^{\mu} O_{\theta(t)}
    &= \bar\lambda^{\mu}
    + c\left( \cos(\theta(t))Z_{i, mn}+ \sin(\theta(t)) \frac{i[O,Z_{mn}]}{2}\right).
\end{align}
In the general case, integrating these rotated operations over time similarly leads to both correctable terms, such as those $\sim \cos(\theta(t)))$ or $\sim \cos(\theta(t)))\sin(\theta(t))$, as well as squared terms $\sim \cos(\theta(t)))^2$, which can be mitigated by adjusting pulse spacings.

\section{Details on pulse sequences for two-mode spin squeezing}

In this section we provide additional details and numerical data about the two-mode spin squeezing generation via subensemble control.

\subsection{Atomic ensemble in optical cavity}

\begin{figure}
    \centering
    \includegraphics[width=\linewidth]{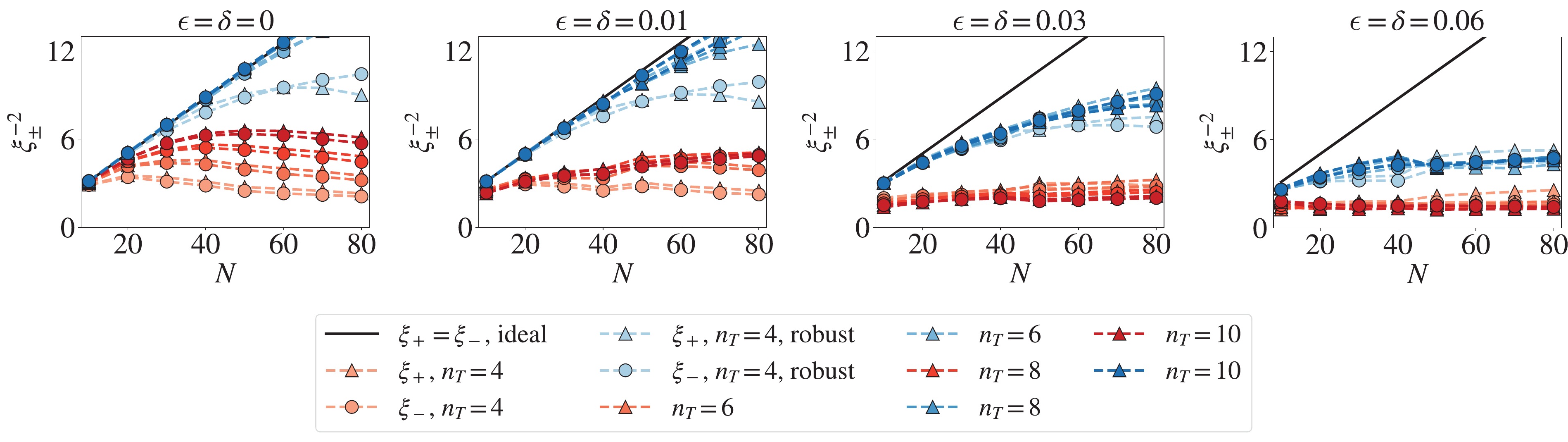}
    \caption{Maximum two-mode spin squeezing as a function of the number of atoms $N$ obtained with a collective two-axis two-spin twisting Hamiltonian. 
    Black lines represent the exact dynamics, marked red lines represent the dynamics under the bare pulse sequence presented in the main text, and blue lines the robust pulse sequence.
    Triangular markers refer to $\xi_+$, round markers refer to $\xi_-$.
    We plot the squeezing parameter for various values of the A-errors $\epsilon$, F-errors $\delta$ and number of cycles $n_T$ of the pulse sequence (including the reflected pulse sequences in the robust case).}
    \label{fig:cavity_SM}
\end{figure}

Ensembles of $N$ atoms in an optical cavity, where two levels in each atom interact with a cavity mode, behave like a collective spin with $N+1$ levels~\cite{Li2022cavity}.
These photon-mediated interactions can be used to realize the native one-axis twisting Hamiltonian
\begin{equation}
    H = \chi (S^z)^2 \ ,
\end{equation}
where $S^z = \sum_{i=1}^N Z_i$ is a collective spin operator.
By detuning the probe laser away from this two-level transition, it is possible to freeze this interaction, thus setting $\chi = 0$.
This allows to alternate the interaction between atoms and the application of pulses, effectively removing PW-errors.
The protocol we envision is therefore one where we alternate pulses and Hamiltonian evolution for short times.
The main source of decoherence in such systems is the global dephasing induced by photon loss from the cavity, which impacts the experiment every time $X$ and $Y$ rotations are applied.
Therefore, we consider the number of pulse sequence cycles $n_T$ as the main experimental cost.

While in the main text we show results for $n_T = 10$, here we also plot $n_T = 6, 8$, see Fig.~\ref{fig:cavity_SM}, and for various values of the A and F errors.
We observe that, even for the case without pulse-level errors, the robust pulse sequence achieves better performance than the bare one.
This is due to the fact that the robust pulse sequence, due to its symmetry, has no contribution from $O(\tau T)$ terms in the Floquet Magnus expansion, thus resulting in $O(\tau^2T)$ errors on top of the average Hamiltonian dynamics.
However, we observe that, already for small values of the errors, the performance of both pulse sequences is largely independent of $\tau$.
In this regime, we can indeed observe that the robust pulse sequences performs consistently better.
In particular, for the values of $N$ considered, the two-mode spin squeezing remains scalable up to $\varepsilon = \delta = 0.03$.

\subsection{Dipolar Rydberg atom array}

\begin{figure}
    \centering
    \includegraphics[width=0.78\linewidth]{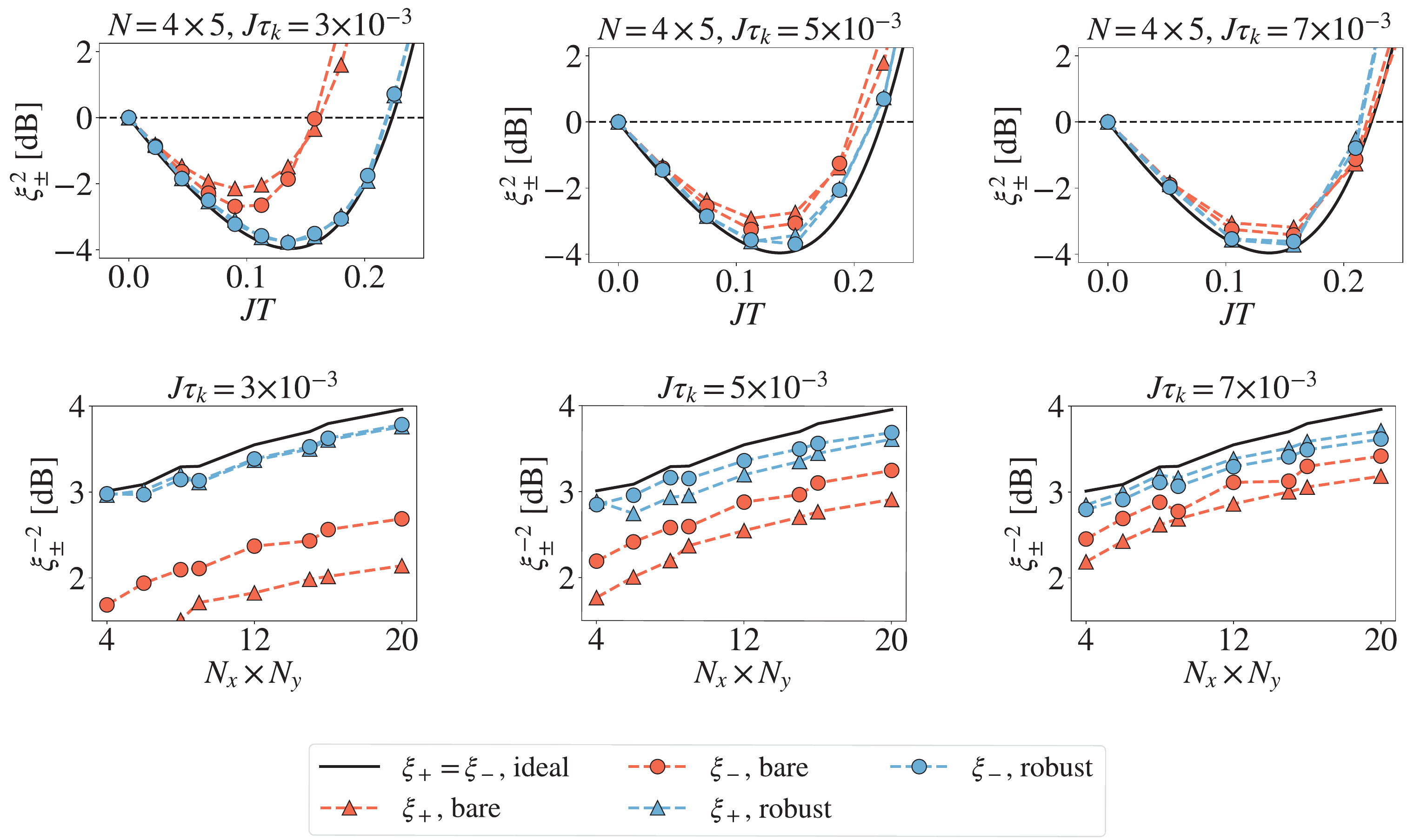}
    \caption{Two-mode spin squeezing as a function of time for a $\alpha = 3$ two-axis two-spin twisting Hamiltonian on a bipartite square lattice, as a function of time $t$ for $N = 20$ sites (upper part) and maximum value as a function of the system size $N$ (lower part).
    Black lines represent the exact dynamics, marked red lines represent the dynamics under the bare pulse sequence presented in the main text, while blue lines the robust pulse sequence.
    Triangular markers refer to $\xi_+$, round markers refer to $\xi_-$.
    We consider A- and F-errors with $\epsilon = \delta = 0.03$, and PW-errors with $\tau_p/\tau = 0.25$, in line with recent experiments.}
    \label{fig:Rydberg_SM}
\end{figure}

Neutral atoms in optical tweezers, when their dynamics is restricted to two Rydberg levels of opposite parity, interact under a long-range spin-exchange Hamiltonian
\begin{equation}
    H = \sum_{i<j}\frac{J}{|i-j|^3}(X_i X_j + Y_i Y_j) \ ,
\end{equation}
where $J\sim a^{-3}$ as a function of the spacing between tweezers $a$; as an example, $J\sim 0.25$~MHz for $a\sim15~\mu$m in the experiment of Ref.~\cite{Bornet2023squeezing}.
In that work, this Hamiltonian has been demonstrated to generate spin squeezing similarly to the collective one-axis twisting Hamiltonian.
Different from the previous case, here interactions between Rydberg levels are always on, so we need to take PW-errors into account.

Let us first comment on the choice of parameters.
In order to generate spin squeezing under the engineered dynamics, we wish to apply at least a few cycles of the pulse sequence before the time at which the maximum squeezing is reached.
From numerical simulations of the two-spin two-axis twisting Hamiltonian with $\alpha = 3$ and $N=20$, this happens at the rescaled total time $JT \simeq 0.2$.
Assuming a pulse width of $25\%$ of the inter-pulse distance, \emph{i.e.} $\tau_p = \tau_k/4$, in line with experimental demonstration~\cite{Bornet2023squeezing,Scholl22Rydberg}, the total cycle time is $7.5\times \tau_k$.
By keeping the minimum duration of pulses fixed, we wish to decrease $J\tau_k$ by decreasing $J$.
This can be done by slightly increasing the lattice spacing.
We take $J\tau_k \in [3\times10^{-3}, 7\times10^{-3}]$, which, assuming $\tau_k = 60~$ns, results in a lattice spacing of at most $a\simeq 25~\mu$m.

The robust pulse sequence is implemented using the reflection explained in the main text.
To further mitigate the PW-errors unaffected by the reflection, we use three additional techniques:
(\emph{i}) after the application of each partially reflected pulse sequence, we switch pulse sequences between $A$ and $B$, and (\emph{ii}) we correct the pulse spacing.
In particular, we choose $\{\tau/6, \tau/6, \tau/6, \tau/6, \tau/6, \tau/6\} = \{\tau/6+\tau_p/2, \tau/6-\tau_p/4, \tau/6, \tau/6, \tau/6-\tau_p/4, \tau/6+\tau_p/2\}$.
With this choice, and by reflecting the two pulse sequences, the resulting average Hamiltonians are 
\begin{equation}
\begin{split}
    &\tilde{H}_{AA,BB}= \frac{4\tau + 7\tau_p}{6(\tau + \tau_p)}\sum_{i<j}\frac{J}{|i-j|^{3}}(X_iX_j+Y_iY_j+Z_iZ_j)\\
    \tilde{H}&_{AB}= \frac{4\tau + 7\tau_p}{6(\tau + \tau_p)}\sum_{i\in A, j\in B}\frac{J}{|i-j|^{3}}\left(\frac{8\tau + 13\tau_p}{8\tau + 14\tau_p} X_iX_j+  Y_iY_j \right) \ .
\end{split}
\end{equation}
For $\tau_p = 0$, the expressions above reduce to the ones in the main text.
While PW-errors introduce an anisotropy between $XX$ and $YY$ interactions, it is generally weak: 
for the value considered $\tau_p/\tau = 1/4$ the anisotropy is $\frac{8\tau + 13\tau_p}{8\tau + 14\tau_p} = 45/ 46$, and it does not appear to affect the numerical results as much as other sources of errors.

In Fig.~\ref{fig:Rydberg_SM} we report additional numerical results about this example, including values of $\tau$ larger than in the main text.
We observe that decreasing $\tau$ increases errors in the bare pulse sequence, while the robust pulse sequence is largely unaffected by errors for our choice of parameters.
Therefore, robust pulse sequences enable the use of higher time resolution to prepare spin squeezed states than the bare pulse sequence, crucial for larger system sizes, where we cannot use exact simulations to calibrate for the maximal generation of spin squeezing.
In this regime, pulse errors could be even more detrimental, as observed in the optical cavity case.
The future study of this regime is further motivated by the numerical evidence of scalable spin squeezing we present in the lower panel of Fig.~\ref{fig:Rydberg_SM}.

\end{document}